 \documentclass[onecolumn,11pt, journal]{IEEEtran}

\IEEEoverridecommandlockouts

\def\hata{\hat{a}}

\def\hatx{\hat{x}}

\def\hatm{\hat{m}}

\global\long\def\EE{\mathbb{E}}
\global\long\def\PP{\mathbb{P}}

\global\long\def\11{\mathbbm{1}}

\def\ulinetau{\underline{\tau}}

\def\ulineY{\underline{B}}

\def\ulinem{\underline{m}}

\def\ulinee{\underline{e}}

\def\ulineX{\underline{X}}
\def\ulinex{\underline{x}}

\def\ulineCalX{\underline{\mathcal{X}}}

\def\ulineCalM{\underline{\mathcal{M}}}

\def\ulineY{\underline{Y}}

\def\ulineX{\underline{X}}
\def\ulineY{\underline{Y}}

\def\ulinelambda{\underline{\lambda}}

\def\ulinehatx{\underline{\hat{x}}}

\def\olinekappa{\overline{\kappa}}

\def\3To1BC{$3-$to$-1$}

\def\define{:{=}~}

\def\naturals{\mathbb{N}}

\def\hatm{\hat{m}}

\newif\ifProofForORDBC

\def\parsec{\par\noindent}
\def\med{\medskip\parsec}

\def\tildea{\tilde{a}}
\def\tildem{\tilde{m}}
\def\CalF{\mathcal{F}}
\def\CalJ{\mathcal{J}}
\def\CalK{\mathcal{K}}
\def\CalH{\mathcal{H}}
\def\CalI{\mathcal{I}}
\def\CalP{\mathcal{P}}

\def\CalM{\mathcal{M}}
\def\CalV{\mathcal{V}}
\def\CalX{\mathcal{X}}

\def\EE{\mathbb{E}}

\def\PP{\mathbb{P}}

\def\11{\mathbbm{1}}

\def\3To1BC{$3-$to$-1$}

\def\define{:{=}~}

\def\naturals{\mathbb{N}}

\def\USB{\mathscr{U}\!\mathcal{S}\!\mathcal{B}-}

\def\TDelta{\mathcal{T}_{\delta}^{(n)}}

\def\pial{\pi_{l}^a}
\def\pimal{\pi_{m_1}^{a,l}}
\def\gammaJoint{\gamma^{a,l}_{m_1}}
\def\gammaXhatU{\gamma_{\hat{x}^n_1}^{u^n}}
\def\gammaXUhat{\gamma_{{x}^n_1}^{\hat{u}^n}}
\def\gammaXhatUhat{\gamma_{\hat{x}^n_1}^{\hat{u}^n}}

\usepackage{stackengine}
\def\deq{\mathrel{\ensurestackMath{\stackon[1pt]{=}{\scriptstyle\Delta}}}}
\def\define{\mathrel{\ensurestackMath{\stackon[1pt]{=}{\scriptstyle\Delta}}}}

\usepackage{stackengine}
\newcounter{relctr} 
\everydisplay\expandafter{\the\everydisplay\setcounter{relctr}{0}} 

\newcommand\labelrel[2]{%
  \begingroup
    \refstepcounter{relctr}%
    \stackrel{\textnormal{(\alph{relctr})}}{\mathstrut{#1}}%
    \originallabel{#2}%
  \endgroup
}

\newif\ifJournal

\usepackage{amssymb}
\usepackage{amsmath}
\usepackage{mathrsfs}
\usepackage{ulem}
\usepackage{epsf,epsfig}
\usepackage{cite}
\usepackage{color}
\usepackage{dsfont}
\usepackage{bbm}
\usepackage{amsthm}
\usepackage{physics}

\newtheorem{theorem}{Theorem}

\newcommand{\comment}[1]{}

\AtBeginDocument{\let\originallabel\label} 

\begin{document}

\sloppy
\newtheorem{remark}{\it Remark}
\newtheorem{thm}{Theorem}
\newtheorem{corollary}{Corollary}
\newtheorem{definition}{Defnition}
\newtheorem{lemma}{Lemma}
\newtheorem{example}{Example}
\newtheorem{prop}{Proposition}

\title{\huge Achievable rate-region for $3-$User Classical-Quantum Interference Channel using Structured Codes}


\author{\IEEEauthorblockN{  Touheed Anwar Atif, Arun Padakandla and S. Sandeep
    Pradhan}\\
\IEEEauthorblockA{Department of Electrical Engineering and Computer Science,\\
University of Michigan, Ann Arbor, MI 48109, USA.\\
Email: \tt touheed@umich.edu, arunpr@utk.edu, pradhanv@umich.edu}
\thanks{This work was supported by NSF grant CCF-2007878.}
}
\maketitle
\thispagestyle{plain}
\pagestyle{plain}

\begin{abstract}


We consider the problem of characterizing an inner bound to the capacity region of a $3-$user classical-quantum interference channel ($3-$CQIC). The best known coding scheme for communicating over CQICs is based on unstructured random codes and employs the techniques of message splitting and superposition coding. For classical $3-$user interference channels (ICs), it has been proven that coding techniques based on coset codes - codes possessing algebraic closure properties - strictly outperform all coding techniques based on unstructured codes. In this work, we develop analogous techniques based on coset codes for $3$to$1-$CQICs - a subclass of $3-$user CQICs. We analyze its performance and derive a new inner bound to the capacity region of $3$to$1-$CQICs that subsume the current known largest and strictly enlarges the same for identified examples.

\end{abstract}

\section{Introduction}
The study of characterizing the \textit{capacity} of a communication medium sheds light on the \textit{structure} of an \textit{optimal coding scheme}. Indeed, the proof of achievability answers questions such as: What structural properties - empirical, algebraic etc. - of a code can enhance its information carrying capability? Can intelligent encoding and decoding rules harness such properties to enhance communication rates? Our findings in this article maybe viewed as providing answers to these questions in the context of classical-quantum (CQ) channels.


We consider the scenario of communicating over a $3-$user classical-quantum interference channel ($3-$CQIC) (Fig.~\ref{Fig:CommnOver3CQIC}). We undertake a Shannon-theoretic study and focus on the problem of 
for characterizing an inner bound to its capacity region. The current known coding schemes for CQICs \cite{sen2012achieving,savov2012network,sen2018inner,hirche2016polar} are based on unstructured codes. In this work, we propose a new coding scheme for a $3-$CQIC based on \textit{nested coset codes} (NCCs) - codes possessing algebraic structure. Analyzing its performance, we derive a new inner bound (Sec.~\ref{Sec:AchReg3to1CQIC}) to the capacity region of $3$to$1-$CQIC - a sub-class of $3-$CQICs. The inner bound is proven to subsume any current known inner bounds based on unstructured codes. Furthermore, we identify examples of $3$to$1-$CQICs for which the derived inner bound is strictly larger. These findings are a first step towards characterizing a new inner bound to the capacity region of a general $3-$CQIC.

The current approach of characterizing the performance limits of CQ channels is based on unstructured codes. 
The reason for this is that for several decades Shannon's
Unstructured codes remained to be the de facto ensemble of codes to be employed in information-theoretic study of any classical channels. Spurred by an ingenious work of K\"orner and Marton \cite{korner1979encode} and followed by findings in a multitude of network communication scenarios \cite{korner1979encode,krithivasan2011distributed,200710TIT_NazGas,200906TIT_PhiZam,201109TITarXiv_JafVis,padakandla2016achievable,padakandla2013computing}, it has been analytically proven that coding schemes designed using codes endowed with algebraic closure properties can strictly outperform all known coding schemes based on unstructured codes. These findings have proven that mere empirical properties are insufficient to achieve performance limits.

The goal of this work is to build on this  and enhance current known coding schemes in the context of CQ channels. Our experience with classical channels suggests that a first step toward this is
 to design and analyze coding schemes for basic building block channels. Indeed, the ensemble of NCCs studied in the simple context of point-to-point (PTP) channels form an important element of this work \cite{2020Bk_PraPadShi}. On the other hand, the mathematical complexity of analyzing CQ channels makes it challenging to generalize even well known coding schemes to the CQ setting. In the light of this, our work maybe viewed as a first step in designing new coding schemes for network CQ channels based on coset codes.

In the context of CQICs, the focus of current research is on $2-$user. There has been considerable effort \cite{fawzi2012classical,sen2012achieving,savov2012network,hirche2016polar,sen2018one,sen2018inner} at proving the achievability of the Han-Kobayashi rate-region (CHK) \cite{198101TIT_HanKob} for $2-$user ICs.
\begin{figure}
    \centering
    \includegraphics[width=3in]{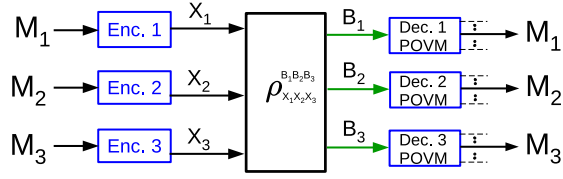}
    \caption{Communication over $3-$CQIC.}
    \label{Fig:CommnOver3CQIC}
\end{figure}
Analogous to these, one can leverage all known coding techniques - message splitting, superposition coding, Marton's binning - and derive an achievable rate region for a $3-$CQIC. See discussion in \cite[Sec.~III]{padakandla2016achievable}. This rate region, henceforth referred to as the $\USB$region contains the largest current known inner bound for any $3-$CQIC. In this work, we focus on $3$to$1-$CQICs (Defn.~\ref{Defn:3To1CQICs}) - a subclass of $3-$IC in which only one receiver (Rx) experiences interference. We propose a coding scheme based on NCCs and derive an inner bound for this sub-class that subsumes the $\USB$region in general, and strictly larger for identified examples (see Ex.~\ref{Ex:NonAdditive3CQIC}). 

To study coset code based coding schemes for basic building block channels, and for pedagogical reasons, we present our findings in two steps. In the first step (Thm. \ref{thm:3to1CQIC}), we demonstrate a construction of a $n$-letter POVM that  can simultaneously decode (i) the correct message and (ii) a bivariate interference component. This first step enables us study performance of NCCs for CQ-PTP channels (Sec.~\ref{Sec:NCCAchieveCQ-PTPCapacity}) and simultaneous decoding of unstructured and NCC codes (Sec.~\ref{Sec:AchReg3to1CQIC}). Our analysis of this simultaneous decoder builds on the technique proposed in \cite{201206TIT_FawHaySavSenWil}. In the next step, we leverage these building blocks and employ a multi-terminal simultaneous decoder \cite{sen2018one} to derive a new achievable rate region for $3$to$1-$CQICs. 
To aid a reader unfamiliar with the central idea of coset codes, we have provided a brief preliminary subsection (Sec.~\ref{sec:prelims}) to convey the utility of algebraic structure.

\section{Preliminaries and Problem Statement}
\label{sec:prelims}
We supplement notation in \cite{2013Bk_Wil} with the following.
For $n\in \mathbb{N}$, $[n] \define \left\{1,\cdots,n \right\}$. 
For a Hilbert space $\mathcal{H}$, $\mathcal{L}(\mathcal{H}),\mathcal{P}(\mathcal{H})$ and $\mathcal{D}(\mathcal{H})$ denote the collection of linear, positive and density operators acting on $\mathcal{H}$ respectively. 
All associated Hilbert spaces are assumed to be finite dimensional. 
We let an \underline{underline} denote an appropriate aggregation of objects. For example, $\ulineCalX \define \mathcal{X}_{1}\times \mathcal{X}_{2} \times \mathcal{X}_{3}$, $\ulinex \define (x_{1},x_{2},x_{3}) \in \ulineCalX$ and in regards to Hilbert spaces $\mathcal{H}_{Y_{i}}: i \in [3]$, we let $\mathcal{H}_{\ulineY} \define \otimes_{i=1}^{3}\mathcal{H}_{Y_{i}}$. We abbreviate the Positive Operator Valued Measure and Block-Length as POVM and B-L, respectively.

Consider a (generic) \textit{$3-$CQIC} $(\rho_{\ulinex} \in \mathcal{D}(\mathcal{H}_{\ulineY}): \ulinex \in \ulineCalX,\kappa_{j}:j \in [3])$ specified through (i) three finite sets $\mathcal{X}_{j}: j \in [3]$, (ii) three Hilbert spaces $\mathcal{H}_{Y_{j}}: j \in [3]$, (iii) a collection of density operators $( \rho_{\ulinex} \in \mathcal{D}(\mathcal{H}_{\ulineY} ): \ulinex \in \ulineCalX )$ and (iv) three cost functions $\kappa_{j}:\mathcal{X}_{j} \rightarrow [0,\infty) : j \in [3]$. The cost function is assumed to be additive, i.e., cost expended by encoder $j$ in preparing the state $\otimes_{t=1}^{n}\rho_{x_{1t}x_{2t}x_{3t}}$ is $\olinekappa_{j}^{n} \define \frac{1}{n}\sum_{t=1}^{n}\kappa_{j}(x_{jt})$. Reliable communication on a $3-$CQIC entails identifying a code. 
\begin{definition}
A \textit{$3-$CQIC code} $c=(n,\ulineCalM,\ulinee,\ulinelambda)$ of B-L $n$ consists of three (i) message index sets $[\mathcal{M}_{j}]: j \in [3]$, (ii) encoder maps $e_{j}: [\mathcal{M}_{j}] \rightarrow \mathcal{X}_{j}^{n}: j \in [3]$ and (iii) POVMs $\lambda_{j} \define \{ \lambda_{j,m}: \mathcal{H}_{j}^{\otimes n} \rightarrow \mathcal{H}_{j}^{\otimes n} : m \in [\mathcal{M}_{j}] \} : j \in [3]$. The average probability of error of the $3-$CQIC code $(n,\ulineCalM,\ulinee,\lambda^{[3]})$ is
\begin{eqnarray}
 \label{Eqn:AvgErrorProb}
 \overline{\xi}(\ulinee,\ulinelambda) \define 1-\frac{1}{\mathcal{M}_{1}\mathcal{M}_{2}\mathcal{M}_{3}}\sum_{\ulinem \in \ulineCalM}\tr(\lambda_{\ulinem}\rho_{c,\ulinem}^{\otimes n}).
 \nonumber
\end{eqnarray}
where $\lambda_{\ulinem} \define \otimes_{j=1}^{3}\lambda_{j,m_{j}}$, $\rho_{c,\ulinem}^{\otimes n} \define \otimes_{t=1}^{n}\rho_{x_{1t}x_{2t}x_{3t}}$ where $(x_{jt}:1\leq t \leq n) = x_{j}^{n}(m_{j}) \define e_{j}(m_{j})$ for $j \in [3]$. Average cost per symbol of transmitting message $\ulinem \in \ulineCalM \in \ulinetau(\ulinee|\ulinem) \define \left( \olinekappa_{j}^{n}(e_{j}(m_{j})): j \in [3] \right)$ and the average cost per symbol of $3-$CQIC code is $\ulinetau(\ulinee) \define \frac{1}{|\ulineCalM|}\sum_{\ulinem \in \ulineCalM}\ulinetau(\ulinee|\ulinem)$.
\end{definition}
\begin{definition}A rate-cost vector $(R_{1},R_{2},R_{3},\tau_{1},\tau_{2},\tau_{3}) \in [0,\infty)^{6}$ is \textit{achievable} if there exists a sequence of $3-$CQIC code $(n,\ulineCalM^{(n)},\ulinee^{(n)},\ulinelambda^{(n)})$ for which $\displaystyle\lim_{n \rightarrow \infty}\overline{\xi}(\ulinee^{(n)},\ulinelambda^{(n)}) = 0$,
\begin{eqnarray}
 \label{Eqn:3CQICAchievability}
 \lim_{n \rightarrow \infty} n^{-1}\log \mathcal{M}_{j}^{(n)} = R_{j}, \mbox{ and }\lim_{n \rightarrow \infty} \ulinetau(\ulinee)_{j} \leq \tau_{j} :j \in [3].
 \nonumber
\end{eqnarray}
\end{definition}
The capacity region $\mathcal{C}(\rho_{\ulinex}:\ulinex \in \mathcal{X})$ of the $3-$CQIC $(\rho_{\ulinex} \in \mathcal{D}(\mathcal{H}_{\ulineY}): \ulinex \in \ulineCalX)$ is the set of all achievable rate-cost vectors.
We define below the sub-class of $3$to$1-$CQICs. 
\begin{definition}
\label{Defn:3To1CQICs}
A $3-$CQIC $(\rho_{\ulinex} \in \mathcal{D}(\mathcal{H}_{\ulineY}): \ulinex \in \ulineCalX)$ is a $3$to$1-$CQIC if (i) for every $\Lambda \in \mathcal{P}(\mathcal{H}_{Y_{2}})$, $\Gamma \in \mathcal{P}(\mathcal{H}_{Y_{3}})$, $\tr((I \otimes \Lambda \otimes I)\rho_{x_{1}x_{2}x_{3}}) = \tr((I \otimes \Lambda \otimes I)\rho_{\hatx_{1}\hatx_{2}\hatx_{3}})$ for every $\ulinex,\ulinehatx \in \ulineCalX$ satisfying $x_{2}=\hatx_{2}$, and (ii) $\tr((I \otimes I \otimes \Gamma)\rho_{x_{1}x_{2}x_{3}}) = \tr((I \otimes I \otimes \Gamma)\rho_{\hatx_{1}\hatx_{2}\hatx_{3}})$ for every $\ulinex,\ulinehatx \in \ulineCalX$ satisfying $x_{3}=\hatx_{3}$.
\end{definition}
\subsection{Illustration of the Central Idea}
\label{Sec:ExAndCentralIdea}
The goal here is to demonstrate the utility of algebraic closure in coding schemes for $3-$ICs. While, we state Ex.~\ref{Ex:Additive3CQIC} in the context of $3$to$1-$CQICs, we discus in the context of a classical $3$to$1-$IC.
The latter provides an exposition on the utility of algebraic closure in network scenarios.
\begin{example}
 \label{Ex:Additive3CQIC}
 Let $\mathcal{X}_{j}=\mathcal{X}=\{0,1\}, \mathcal{H}_{j}=\mathcal{H}$, $\sigma^{(j)}_{b} \in \mathcal{D}(\mathcal{H})$ for $j \in [3]$ and $b \in \mathcal{X}$. For $\ulinex \in \ulineCalX$, let 
$\rho_{\ulinex} \deq \sigma^{(1)}_{x_{1}\oplus x_{2} \oplus x_{3}} \otimes \sigma^{(2)}_{x_{2}} \otimes \sigma^{(3)}_{x_{3}}$.
 For $x \in \{0,1\}$, we let $\kappa_{1}(x)=x$ and $\kappa_{k}(x)=0$ for $k =2,3$.
\end{example}
Let $\mathcal{H} = \mathbb{C}^{2}$, $\sigma_{b}(\eta) \define (1-\eta)\ket{b}\bra{b}+\eta \ket{1-b}\bra{1-b}$ for $b \in \mathcal{X}$, $\eta \in [0,1]$. Let $\sigma_{b}^{(1)}\define \sigma_{b}(\delta_{1})$ and $\sigma^{(2)}_{b}\define \sigma^{(3)}_{b}\define \sigma_{b}(\delta)$ for $b \in \mathcal{X}$ and some specified $\delta, \delta_1 \in (0,1)$ . In addition, let $\tau \in (0,\frac{1}{2})$ specify a Hamming cost constraint on Tx $1$'s input. With this choice, one identifies the above example with a $3$to$1-$IC $Y_{1}= X_{1} \oplus X_{2} \oplus X_{3} \oplus N_{1}$, $Y_{k} = X_{k} \oplus N_{k} : k=2,3$ with $N_{1}\sim $ Ber$(\delta_{1})$, $N_{k} \sim$ Ber$(\delta)$ $k=2,3$ being independent. 
Tx $k \in \{2,3\}$ splits its information into $U_{k},X_{k}$. Rx $1$ decodes $U_{2},U_{3},X_{1}$, while Rx $k \in \{2,3\}$ decodes $U_{k},X_{k}$. So long as $H(U_{k}|X_{k}) > 0$ for either $k \in \{2,3\}$, it can be shown that $H(X_{2}\oplus X_{3}|U_{2},U_{3}) > 0$ implying Tx-Rx $1$ cannot achieve $h_{b}(\delta_{1}*\tau)-h_{b}(\delta_{1})$ - its interference free cost constrained capacity. If $h_{b}(\delta_{1}*\tau)-h_{b}(\delta_{1}) + 2(1-h_{b}(\delta)) > 1-h_{b}(\delta_{1})$, it can be shown that $H(U_{k}|X_{k}) > 0$ for either $k \in \{2,3\}$ precluding Tx-Rx $1$ achieving a rate $h_{b}(\delta_{1}*\tau)-h_{b}(\delta_{1})$ using unstructured coding. 
Suppose users $2,3$ employ codes of rate $1-h_{b}(\delta)$ that are cosets of the \textit{same linear code}, then the above condition does not preclude Tx-Rx $1$ from achieving a rate  $h_{b}(\delta_{1}*\tau)-h_{b}(\delta_{1})$, so long as $\tau * \delta < \delta$, even if $1+h_{b}(\tau*\delta_{1}) > 2h_{b}(\delta)$. 
The reason is, user $2$ and $3$'s codebooks when added is another coset of the same rate $1-h_{b}(\delta)$. Rx $1$ can just decode this interference if $h_{b}(\delta_{1}*\tau)-h_{b}(\delta_{1})+1-h_{b}(\delta) < 1- h_{b}(\delta_{1})$ which is equivalent to $\tau * \delta < \delta < \frac{1}{2}$. 
Hence, for this $3$to$1-$IC, if $h_{b}(\delta_{1}*\tau)-h_{b}(\delta_{1}) + 2(1-h_{b}(\delta)) > 1-h_{b}(\delta_{1})$ and $\tau * \delta < \delta < \frac{1}{2}$ hold, then coset codes are strictly more efficient than unstructured codes. 

\section{Rate region using Coset Codes for 3to1-CQIC} 
\label{Sec:AchReg3to1CQIC}
In this section we consider the above described 3to1-CQIC and provide an achievable rate-region.
\begin{theorem} \label{thm:3to1CQIC}
Given a $3to1$-CQIC $(\rho_{\ulinex} \in \mathcal{D}(\mathcal{H}_{\ulineY}): \ulinex \in \ulineCalX, \kappa_j:j\in[3])$ and a PMF $p_{V_2V_3X_1X_2X_3} = p_{X_1}p_{V_2X_2}p_{V_3X_3}$ on $\mathcal{V}_2\cross\mathcal{V}_3\cross\mathcal{X}_2\cross\mathcal{X}_3$ where $\mathcal{V}_2 = \mathcal{V}_3 = \mathcal{F}_q$, a rate-cost triple $(R_1,R_2,R_2,\tau_1,\tau_2,\tau_3)$ is achievable if it satisfies the following 
\begin{align}
    R_1 &\leq I(Y_1;X_1|U)_{\sigma_1},\quad 
    R_j  \leq I(Y_j;V_j)_{\sigma_2},\nonumber \\
    R_j &\leq \min\{H(V_2),H(V_3)\} - H(U) + I(Y_1;U|X_1)_{\sigma_1},\nonumber\\
    R_1 + R_j &\leq \min\{H(V_2),H(V_3)\} - H(U) + I(Y_1;V_1U)_{\sigma_1}, \nonumber 
\end{align}
for $j=2,3,$ and $\EE[\kappa_{j}(X_j)] \leq \tau_j: j\in[3]$,  where
\begin{align}
    \sigma_1^{\ulineY}&\deq  \sum_{x_1\in\mathcal{X}_1,u \in \CalF_q}p_{X_1}(x_1)p_{U}(u)\rho_{x_1,u}^{\ulineY}\otimes \ketbra{x_1}\otimes \ketbra{u},\nonumber \\
    \rho_{x_1,u}^{\ulineY} &\deq \sum_{v_2,v_3}\sum_{x_2,x_3}p_{V_2,V_3,X_2,X_3|U}(v_2,v_3,x_2,x_3|u)\rho_{\ulinex}^{\ulineY} \nonumber \\
    \sigma_2 &= \!\!\sum_{v_1,v_2,v_3}\! p_{\ulineX V_2V_3}(\ulinex,v_2,v_3)\rho_{\ulinex}^{\ulineY} \otimes \ketbra{v_2}\otimes\ketbra{v_3}, \nonumber
\end{align}for $U \deq V_2 \oplus V_3,$ and $\{\ket{v_2}\}$, $\{\ket{v_3}\}$ as some 
orthonormal basis on $\CalH_Y$.
\end{theorem}
\begin{example}
 \label{Ex:NonAdditive3CQIC}
 Let $\mathcal{X}_{j}=\mathcal{X}=\{0,1\}, \mathcal{H}_{j}=\mathbb{C}^2$ and let \\
 \begin{align}
 \sigma_{0}\deq\left[
  \begin{array}{cc} 2/3 &  0 \\ 0 & 1/3 \end{array} \right], \mbox{ and }
 \sigma_1\deq\left[
  \begin{array}{cc} 1/2 &  1/6 \\ 1/6 & 1/2 \end{array} \right].\nonumber
\end{align}
 \begin{align*} 
 \label{Eqn:Add3CQICExDensOpDesc}
  \rho_{\underline{x}} &\deq  [ (1-\delta_1)\sigma_{x_{1}\oplus x_{2} \oplus x_{3}} +\delta_1
\sigma_{x_{1}\oplus x_{2} \oplus x_{3}
\oplus 1}]  \otimes[(1-\delta) \sigma_{x_{2}} + \delta\sigma_{x_{2}\oplus 1}]
 \otimes [(1-\delta)\sigma_{x_{3}}+ \delta\sigma_{x_{3}\oplus 1}], \nonumber
 \end{align*}
 for $\ulinex \in \ulineCalX,$  where $N_1$, $N_2$ and $N_3$ are mutually independent Bernoulli random variables with biases
 $\delta_1,\delta$ and $\delta$, respectively. We let $\delta_1,\delta \in (0,0.5)$. 
 For $x \in \{0,1\}$, we let $\kappa_{1}(x)=x$ and $\kappa_{k}(x)=0$ for $k =2,3$.
Let $\rho(p) := p\sigma_0+(1-p)\sigma_1$.
Note that $\rho(p)$ and $\rho(1-p)$ do not 
commute except for $p=0.5$. 
 It can be checked that $S(\rho(p))$ is a symmetric concave function of $p \in (0,1)$.
 Consider the case when $\tau * \delta_1 \leq \delta$.
 Using NCC, the three users can achieve their PTP capacities simultaneously:
 $S(\rho(\tau * \delta_1))-S(\rho(\delta_1))$, $S(\rho(0.5))-S(\rho(\delta))$, and
 $S(\rho(0.5))-S(\rho(\delta))$, respectively. These correspond to the rates given by
 $I(X_1;B_1|X_2 \oplus X_3)$, $I(X_2;Y_2)$, and $I(X_3;Y_3)$.
 One can show that if
$ S(\rho(\tau *\delta_1))  -S(\rho(\delta_1))+2(S(\rho(0.5))-S(\rho(\delta)))  > S(\rho(0.5))-S(\rho(\delta_1))$,
 then using unstructured codes, all three users cannot achieve their respective capacities
 simultaneously. This condition is equivalent to the condition:
$S(\rho(\tau * \delta_1))+S(\rho(0.5)) >2S(\rho(\delta))$.
Hence by choosing $\tau * \delta_1=\delta$, and $\delta<0.5$, we see that NCC-based coding scheme enables all users achieve their respective capacities simultaneously, while this is not possible in unstructured coding scheme.
 
 \end{example}
 
\begin{proof}
We divide the proof into three parts entailing the encoding, decoding and error analysis techniques.
\subsection{Encoding Technique}
Consider a PMF $p_{V_2V_3\ulineX}$ on $\CalV_2\times\CalV_3\times\ulineCalX$ with $\CalV_2 = \CalV_3 = \CalF_q$, and choose $n $ and $R_j: j=[3]$ as non-negative integers. For encoder 1, we use the random coding strategy and construct a codebook $\mathcal{C}_1 \deq \{x_1(m_1): m_1 \in [2^{nR_1}]\}$ on $\mathcal{X}_1$ using the marginal PMF $p_{X_1}^n$. Let $e_1(m_1) \deq x_1(m_1): m_1 \in [2^{nR_1}] $ denote this encoding map. However, to construct the codebooks for encoders 2 and 3, we employ a technique based on nested coset codes. Since, the structure and encoding rules are identical for the these two encoders, we describe it using a generic index $j \in \{2,3\}.$ Let $e_j: \CalF_q \rightarrow \CalX_j^n: j=1,2$ denote the encoding maps. We define an NCC as
follows. 
\begin{definition}
 \label{Defn:NCC}
 An $(n,k,l,g_{I},g_{O/I},b^{n})$ NCC built over a finite field $\mathcal{V}=\mathcal{F}_{q}$ comprises of (i) generator matrices $g_{I} \in \mathcal{V}^{k \times n}$, $g_{O/I} \in \mathcal{V}^{l \times n}$ (ii) a dither/bias vector $b^{n}$, an encoder map $e :\mathcal{V}^{l} \rightarrow \mathcal{V}^{k}$. We let $v^{n}(a,m) = ag_{I}\oplus_{q}mg_{O/I}\oplus_{q}b^{n}: (a,m) \in \mathcal{V}^{k} \times \mathcal{V}^{l}$ denote elements in its range space.
\end{definition}
Consider two NCCs with parameters $(n,k,l,g_{I},g_{O/I},Y_j^n) : j\in \{2,3\}$ defined using the above definition,
with their range spaces denoted by $v_j^n(a_j,m_j) : j \in \{2,3\},$ respectively. Note that the choice of $g_{I}$ and $g_{O/I}$ are identical for the two NCCs. Further, let $\theta_j(m_j) \deq \sum_{a_j \in \CalF_q^k}\11_{\{v_j^n(a_j,m_j) \in \TDelta(p_{V_j})\}}$. For every message $m_j$ the encoder $j$ looks for a codeword in the coset $v_j^n(a_j,m_j) :a_j \in \mathcal{F}_q^k $ that is typical according to $p_{V_j}$. If it finds at least one such codeword, one of them,  say $v_j^n(\alpha_j(m_j),m_j),$ is chosen randomly and uniformly. $e_j(m_j)$ is generated according to $p_{X_j|V_j}^n(\cdot |v_j^n(\alpha_j(m_j),m_j))$ and is transmitted on the CQIC. Otherwise, if it finds none in the coset that is typical according to $p_{V_j}$,, and error is declared. This specifies the encoding rule for the three encoders. Now we describe the decoding rule.

\subsection{Decoding Description}
We begin the describing first decoder. Unlike a generic $3$ CCIC decoding technique of recovering the three messages, the decoder here is constructs its POVM to recover its own message and only a bi-variate function of the two interfering messages. Since, the POVMs here require joint typicality of two messages, we employ the POVM construction similar to \cite{fawzi2012classical}, while incorporating the bi-variate function being decoded.
For this, we equip the decoder 1 with the NCC $(n,k,l,g_{I},g_{O/I},b^n),$ where $b^n = b_1^n\oplus Y_2^n.$. We define $u^n(a,l)$ as the range space of the above NCC. Toward specifying the decoding POVM, we let $\pi_{m_1}\deq \pi_{x_1^n(m_1)}
$, $\pi_{a,l} \deq \pi_{u^n(a,l)}\11_{\{u^n(a,l) \in \TDelta(p_U)\}}$, $\pi_{m_1}^{a,l}\deq \pi_{x_1^n(m_1),u^n(a,l)}\11_{\{(x_1^n(m_1),u^n(a,l)) \in \TDelta(p_{X_1U})\}}$,  denote the conditional typical projectors (as defined in \cite[Def. 15.2.4]{2013Bk_Wil}) with respect to the states 
$ \rho_{x_1}^{Y_1} \deq \sum_{u}p_{U}(u)\rho_{x_1,u}^{Y_1}$, $ \rho_{u}^{Y_1} \deq \sum_{x_1} p_{X_1}(x_1) \rho_{x_1,u}^{Y_1}$ and $ \rho_{x_1,u}^{Y_1}$
, respectively, where $ \rho_{x_1,u}^{Y_1}$ is as defined in the theorem statement. In addition, let $\pi_{\rho}^{Y_1}$ denote the typical projector with respect to the state $\rho \deq \sum_{x_1,u}p_{X_1}(x_1)p_{U}(u)\rho_{x_1,u}^{Y_1}$. Using these projectors, we define the POVM $\lambda_{\CalI_1}^{Y_1} \deq \{\lambda^{Y_1}_{m_1,a,l}\}$, where
\begin{align}
    \lambda^{Y_1}_{m_1,a,l} \deq \Big(\sum_{\substack{\hat{m}_1 \in \\ [2^{nR_1}]}}\sum_{\substack{\hat{a}\in \CalF_q^k\\\hat{l}\in \CalF_q^l }}\gamma^{\hat{a},\hat{l}}_{\hat{m}_1}\Big)^{-1/2}\hspace{-5pt}\gamma_{m_1}^{a,l}    \Big(\!\sum_{\substack{\hat{m}_1 \in \\ [2^{nR_1}]}}\sum_{\substack{\hat{a}\in \CalF_q^k\\\hat{l}\in \CalF_q^l }}\gamma^{\hat{a},\hat{l}}_{\hat{m}_1}\Big)^{-1/2}, \nonumber
\end{align}
$\lambda_{-1} \deq I - \sum_{{m_1}\in [2^{nR_1}]}\sum_{{a} \in \CalF_q^k}\sum_{{l} \in \CalF_q^l} \lambda^{Y_1}_{m_1,a,l}$ and $ \gamma^{a,l}_{m_1} \deq \pi_{\rho}\pi_{m_1}\pi_{m_1}^{a,l}\pi_{m_1}\pi_{\rho}$.
Having described the first decoder, we move on to describing the other two. Since these two decoders are identical, we use a generic variable $j$ to refer to each of these. We define $\pi_{\rho}^{j}$ and $\pi_{a_j,m_j}^{j}$ as the typical \cite[Def. 15.1.3]{2013Bk_Wil} and the conditional typical projectors \cite[Def. 15.2.4]{2013Bk_Wil} with respect to the states $\rho^{Y_j} \deq \sum_{v_j}p_{V_j}(v_j)\rho_{v_j}^{Y_j}$ and $\rho_{v_j}^{Y_j},$ respectively. Using this, we construct the POVM $\lambda_{\CalI_j}^{Y_j} \deq \{\lambda^{Y_j}_{m_j,a_j}\}$, for encoder $j$ as 
\begin{align}
\lambda_{a_j,m_j}^{Y_j}\! &\deq\! \Big(\!\!\sum_{\hat{a_j} \in \CalF_q^k}\!\sum_{\hat{m_j}\in \CalF_q^l}\!\!\!\zeta_{\hat{a_j},\hat{m_j}}\!\Big)^{-1/2}\hspace{-15pt}\zeta_{a_j,m_j}    \Big(\!\!\sum_{\hat{a_j} \in \CalF_q^k}\!\sum_{\hat{m_j}\in \CalF_q^l}\!\!\!\zeta_{\hat{a_j},\hat{m_j}}\!\Big)^{-1/2}, \nonumber
\end{align}
$\lambda_{-1}^{Y_j}\define I-\sum_{m \in \CalV^{l}}\sum_{a \in \CalV^{k}}\lambda_{a,m}^{Y_j}$ and $\zeta_{a_j,m_j} \define \pi_{\rho}^{j}\pi_{a_j,m_j}^{j}\pi_{\rho}^{j}$. Lastly, we provide the distribution of the random NCC.\\
\paragraph*{Distribution of the Random Coset Code} : The objects $g_{I}\in \mathcal{V}^{k \times n},g_{O/I} \in \mathcal{V}^{l \times n},b^{n} \in \mathcal{V}^{n}$ and the collection $(a_{m} \in s(m): m \in \mathcal{V}^{l})$ specify a NCC CQ-PTP code unambiguously. A distribution for a random code is therefore specified through a distribution of these objects. We let upper case letters denote the associated random objects, and obtain
\begin{eqnarray}
 \CalP\!\left( \begin{array}{c} G_{I}=g_{I},G_{O/I}=g_{O/I}\\B^{n}_j=b^{n}_j,\alpha_j(m_j)=a_{j}: m_j \in \CalF_q^l \end{array} \right) \!= q^{-(k+l+1)n}\prod_{m \in \CalF_q^{l}}\frac{1}{\Theta_j(m_j)}.\nonumber
\end{eqnarray}
\subsection{ Error Analysis}
As in a general information theoretic setting, we derive upper bounds on probability of error $\overline{\xi}(\ulinee,\ulinelambda) $ by averaging over the random code of the first user and the ensemble of nested coset codes used by the other two users. The error probability of this code is given by
\begin{align}
 \overline{\xi}(\ulinee,\ulinelambda) \define 1-\frac{1}{\mathcal{M}_{1}\mathcal{M}_{2}\mathcal{M}_{3}}\sum_{\ulinem \in \ulineCalM}\tr(\lambda_{\ulinem}^{\ulineY}\rho_{c,\ulinem}^{\otimes n}). 
\end{align} 
Using the  inequality 
\begin{align}
    \displaystyle (I - \lambda_{\uline m}^{\ulineY}) \leq (I - \lambda_{m_1}^{Y_1})\otimes I^{Y_2Y_3} + (I - \lambda_{m_1}^{Y_2})\otimes I^{Y_1Y_3} + (I - \lambda_{m_1}^{Y_3})\otimes I^{Y_1Y_2},
\end{align}
from \cite{abeyesinghe2009mother}, we get  $\overline{\xi}(\ulinee,\ulinelambda) \leq S_1 +S_2 + S_3,$ where 
\begin{align}
    S_j \deq \frac{1}{\underline{\mathcal{M}}}\sum_{\ulinem \in \ulineCalM}\tr(\left((I-\lambda_{m_j}^{Y_j})\otimes I^{\ulineY\backslash B_i}\right) \rho_{c,\ulinem}^{\otimes n}): j \in [3].\nonumber
\end{align}
Using the definition of $3to1$-CQIC, we can further simplify $S_2$ and $S_3$ as
\begin{align}
    S_j = \frac{1}{\CalM_j}\sum_{m_j}\tr((I-\lambda_{m_j}^{Y_j})\rho_{e(m_j)}) : j\in \{2,3\}. \nonumber
\end{align}

We first consider the terms $S_2, S_3$
. Note that, due to the nature of the $3$to$1-$CQIC problem definition, the terms $S_2$ and $S_3$ are identical to a point-to-point (PTP) setup. Therefore, to bound these terms we construct a CQ-PTP problem setup in the sequel (see Sec.~\ref{Sec:NCCAchieveCQ-PTPCapacity}) 
and employ that as a module in bounding $S_2, S_3$. The following proposition formalizes this.


\begin{prop}\label{prop:Lemma for S_23}
There exists  $\epsilon_{S}(\delta), \delta_{S}(\delta),$ such that for  all 
$\delta$ and sufficiently large $n$, we have $\EE\left[S_2+S_3\right] \leq \epsilon_{{S}}(\delta) $, if $ R_j \leq  I(Y_j;V_j )_{\sigma_2}  + \delta_{S} : j= 2,3$, 
where   
$\epsilon_{{S}},\delta_{S} \searrow 0$ as $\delta \searrow 0$.
\end{prop}
\begin{proof}
The proof is provided in Section \ref{Sec:NCCAchieveCQ-PTPCapacity}.
\end{proof}

Now, we move on to bounding the term $S_1$. Let $\mathscr{E} \deq \{\theta_1(m_1) =0 \text{ or } \theta_2(m_2) = 0\}$. By noting that $S_1 \leq 1$, we obtain $S_1 \leq S_1'  + \11_{\mathscr{E}}$, where $S_1' \deq S_1\cdot \11_{\mathscr{E}^c}$. As a first step, we bound the indicator $\11_{\mathscr{E}}$ using the following proposition.
\begin{prop}\label{prop:PTP:Lemma for E}
There exist  $\epsilon_{E}(\delta), \delta_{E}(\delta),$ such that for  all 
$\delta$ and sufficiently large $n$, we have $\EE_{\CalP}\left[\mathscr{E}\right] \leq \epsilon_{{E}}(\delta) $, if $ \frac{k}{n} \geq \log{q} - \min\{H(V_1),H(V_2)\} + \delta_E$, where  $\epsilon_{{E}},\delta_{E} \searrow 0$ as $\delta \searrow 0$.
\end{prop}
\begin{proof}
The proof follows from \cite[App.~B]{MACwithStates_ArunPad_SandeepPra}.
\end{proof}
Now considering the term $S_1'$, and using the linearity of trace while ignoring some negative terms, we get
\begin{align}\label{eq:boundS_1'}
      S_1' &\leq \frac{1}{\underline{\mathcal{M}}}\sum_{\ulinem \in \ulineCalM}\tr(\Big((I-\lambda_{m_1,a,l}^{Y_1})\otimes I^{Y_2Y_3}\Big) \rho_{c,\ulinem}^{\otimes n})\11_{\mathscr{E}^c} \nonumber \\
      &\leq \frac{1}{\underline{\mathcal{M}}}\sum_{\ulinem \in \ulineCalM}{\tr((I-\lambda_{m_1,a,l}^{Y_1}) \pimal\rho_{c,\ulinem}^{Y_1}\pimal)}\11_{\mathscr{E}^c}  + S_{11},
\end{align}
where the second inequality defines the following $S_{11} \deq \left\|\pimal\rho_{c,\ulinem}^{Y_1}\pimal - \rho_{c,\ulinem}^{Y_1} \right\|_1, \rho_{c,\ulinem}^{Y_1}\deq  \tr_{Y_2Y_3}(\rho_{c,\ulinem}^{\otimes n})$, $a \deq \alpha_1(m_1) \oplus \alpha_2(m_2)$, and $ l \deq  m_1 \oplus m_2$ and uses the inequality  $\tr(\lambda\rho) \leq \tr(\lambda\sigma) + \left\|\rho-\sigma\right\|_1$ which holds for all $0 \le \rho,\sigma,\lambda \leq 1$. 
Before we begin the proof, we provide the following lemma based on the pinching for non-commutating operators \cite{2013Bk_Wil,sutter2018approximate}. 
\begin{lemma}\label{lem:LemmaPinching}
For $\pimal, \pi_{m_1}, \pi_l^a $ and $\pi_{\rho} $ as defined above, we have  $$\mbox{tr}(\pimal\rho_{c,\ulinem}^{Y_1})\geq 1- \epsilon_{p_1}(\delta), \quad \mbox{tr}(\pi_{m_1}\rho_{c,\ulinem}^{Y_1})\geq 1- \epsilon_{p_2}(\delta), \quad \mbox{tr}(\pi_l^a\rho_{c,\ulinem}^{Y_1})\geq 1- \epsilon_{p_3}(\delta),  \quad \mbox{tr}(\pi_\rho\rho_{c,\ulinem}^{Y_1}) \geq 1- \epsilon_{p_4}(\delta),$$ where $\epsilon_{p_i}(\delta):i\in[4] \searrow 0$ as $\delta \searrow 0.$
\end{lemma}
\begin{proof}
The proof is provided in Appendix \ref{appx:proofLemmaPinching}.
\end{proof}
Using the above lemma, we first bound the term corresponding to $S_{11}.$ Applying the gentle operator lemma \cite[Lem. 9.4.2]{2013Bk_Wil} on  $S_{11}$, we obtain
\begin{align}
    S_{11} \deq \left\|\pial\rho_{c,\ulinem}^{Y_1}\pial - \rho_{c,\ulinem}^{Y_1} \right\|_1 \leq 2\sqrt{1-\tr(\pial\rho_{c,\ulinem}^{Y_1})} \leq 2\sqrt{\epsilon_p(\delta)}.
\end{align}

 Considering the first term in the right hand side of \eqref{eq:boundS_1'}, let $T$ denote a generic term within its summation, defined as
\begin{align}
    T \deq  {\tr((I-\lambda_{m_1,a,l}^{Y_1}) \pial\rho_{c,\ulinem}^{Y_1}\pial)}\11_{\mathscr{E}^c}. \nonumber
\end{align}
This term can be bounded using the Hayashi-Nagaoka inequality \cite{2013Bk_Wil} as $T \leq 2(1-T_1) + 4T_2$, where 
\begin{align}
    T_1 \deq \tr(\gammaJoint \pial\rho_{c,\ulinem}^{Y_1}\pial), \qquad T_2 \deq\sum_{\substack{(m_1',a',l') \neq (m_1,a,l)} }\tr(\gamma_{m_1'}^{a',l'} \pial\rho_{c,\ulinem}^{Y_1}\pial).\nonumber
\end{align}
Note the objective now is to prove $T_1$ is close to one and $T_2$ is close to zero. Consider the following proposition in regards to $T_1$.
\begin{prop}\label{prop:Lemma for T_1}
There exist  $\epsilon_{T_1}(\delta), \delta_{T_1}(\delta),$ such that for  all sufficiently small $\delta$ and sufficiently large $n$, we have $\EE\left[T_1\right] \geq 1- \epsilon_{{T_1}}(\delta) $, where  $\epsilon_{{T_1}},\delta_{T_1} \searrow 0$ as $\delta \searrow 0$.
\end{prop}
\begin{proof}
Using $\tr(\lambda\rho) \geq \tr(\lambda\sigma) - \left\|\rho-\sigma\right\|_1$ for $ 0 \leq \rho,\sigma,\lambda \leq I$, we have
\begin{align}
    T_1 &\geq \tr(\pimal\rho_{c,\ulinem}^{Y_1}) -\left\|\pi_\rho\rho_{c,\ulinem}^{Y_1}\pi_\rho - \rho_{c,\ulinem}^{Y_1}\right\| - \left\|\pi_l^a\rho_{c,\ulinem}^{Y_1}\pi_l^a - \rho_{c,\ulinem}^{Y_1}\right\| - \left\|\pi_{m_1}\rho_{c,\ulinem}^{Y_1}\pi_{m_1} - \rho_{c,\ulinem}^{Y_1}\right\| \nonumber\\
    & \geq \tr(\pimal\rho_{c,\ulinem}^{Y_1}) - 2\sqrt{1-\tr(\pi_\rho\rho_{c,\ulinem}^{Y_1})} - 2\sqrt{1- \tr(\pi_l^a\rho_{c,\ulinem}^{Y_1})} -2\sqrt{1-\tr(\pi_{m_1}\rho_{c,\ulinem}^{Y_1})} \nonumber \\
    &\geq 1 - \epsilon_{T_1}(\delta),
\end{align}
where the second inequality follows from the gentle opertor lemma and the last inequality uses the above Lemma \ref{lem:LemmaPinching} by defining $\epsilon_{T_1} \deq \epsilon_{p_1} + 2(\sqrt{\epsilon_{p_2}} + \sqrt{\epsilon_{p_3}} + \sqrt{\epsilon_{p_4}})$. This completes the proof.
\end{proof}

Now, we move on to bounding the term $T_2$. Firstly, note that the summation in $T_2$ can be split into seven different summations based on how many indices within the summation over the triple $(m_1',a',l')$ are equal to  $(m_1,a,l)$. However, only three of these seven provide binding constraints on the rate triple $(R_1,R_2,R_3)$. Building on this  we perform the split $T_2 = T_{21} + T_{22} + T_{23} + T_{3},$ where 
\begin{align}
    T_{21} \deq \sum_{m_1'\neq m_1}\tr(\gamma_{m_1'}^{a,l} \pial\rho_{c,\ulinem}^{Y_1}\pial), \quad T_{22} \deq \sum_{\substack{a'\neq a, l'\neq l}}\tr(\gamma_{m_1}^{a',l'} \pial\rho_{c,\ulinem}^{Y_1}\pial),\quad 
  T_{23} \deq \sum_{\substack{m_1' \neq m_1,\\a'\neq a, l'\neq l}}\tr(\gamma_{m_1'}^{a',l'} \pial\rho_{c,\ulinem}^{Y_1}\pial)\nonumber
\end{align}
 represents the rate constraining (binding) terms while $T_{3} \deq T_2 - \sum_{i=1}^{3}T_{2i}$ represents the inactive terms (with respect to constraining the rate). 
 We provide the following set of propositions bounding each of these terms $T_{2i}: i\in[3]$.
 \begin{prop}\label{prop:LemmaT_21}
There exists  $\epsilon_{T_{21}}(\delta), \delta_{T_{21}}(\delta),$ such that for  all sufficiently small $\delta$ and sufficiently large $n$, we have $\EE\left[T_{21}\right] \leq \epsilon_{{T_{21}}}(\delta) $ if $ R_1 + \frac{2k}{n}\log{q} \leq {2\log{q}} -H(V_1,V_2) +  I(Y_1;X_1|U)_{\sigma_1}  + \delta_{T_{21}}$, where  $\epsilon_{{T_{21}}},\delta_{T_{21}} \searrow 0$ as $\delta \searrow 0$.
\end{prop}
\begin{proof}
The proof is provided in Appendix \ref{appx:proofofPropT_21}
\end{proof}
Now, we provide the proposition for $T_{22}$ as follows.
 \begin{prop}\label{prop:LemmaT_22}
There exists  $\epsilon_{T_{22}}(\delta), \delta_{T_{22}}(\delta),$ such that for  all sufficiently small $\delta$ and sufficiently large $n$, we have $\EE\left[T_{22}\right] \leq \epsilon_{{T_{22}}}(\delta) $ if $  \frac{3k+l}{n}\log{q} \leq {3\log{q}} -H(V_1,V_2) -H(U) +  I(Y_1;U|X_1)_{\sigma_1}  + \delta_{T_{22}},$ where  $\epsilon_{{T_{22}}},\delta_{T_{22}} \searrow 0$ as $\delta \searrow 0$.
\end{prop}
\begin{proof}
The proof is provided in Appendix \ref{appx:proofofPropT_22}
\end{proof}
 \begin{prop}\label{prop:LemmaT_23}
There exists  $\epsilon_{T_{23}}(\delta), \delta_{T_{23}}(\delta),$ such that for  all sufficiently small $\delta$ and sufficiently large $n$, we have $\EE\left[T_{23}\right] \leq \epsilon_{{T_{23}}}(\delta) $ if $ R_1 + \frac{3k+l}{n}\log{q} \leq {3\log{q}} -H(V_1,V_2) - H(U) +  I(Y_1;X_1,U)_{\sigma_1}  + \delta_{T_{23}},$ where  $\epsilon_{{T_{23}}},\delta_{T_{23}} \searrow 0$ as $\delta \searrow 0$.
\end{prop}
\begin{proof}
The proof is provided in Appendix \ref{appx:proofofPropT_23}
\end{proof}
For the terms in the expression $T_{3}$,  we split $T_{3}$ as $T_3\deq T_{31} + T_{32}+T_{33}+T_{34}$, where
\begin{align}
    T_{31}\deq \sum_{a'\neq a}\tr(\gamma_{m_1}^{a',l} \pial\rho_{c,\ulinem}^{Y_1}\pial), &\quad T_{32} \deq \sum_{l'\neq l}\tr(\gamma_{m_1}^{a,l'} \pial\rho_{c,\ulinem}^{Y_1}\pial),  \nonumber \\ 
    T_{33} \deq \sum_{\substack{m_1'\neq m_1,\\a'\neq a}}\tr(\gamma_{m_1'}^{a',l} \pial\rho_{c,\ulinem}^{Y_1}\pial),&\quad  T_{34} \deq \sum_{\substack{m_1'\neq m_1,\\l'\neq l}}\tr(\gamma_{m_1'}^{a,l'} \pial\rho_{c,\ulinem}^{Y_1}\pial).
\end{align} 
As mentioned earlier, the analysis of these term follows from the analysis performed for the terms $T_{2i}, i \in [3],$ and further these terms  do not contribute to any new additional rate constraints. However, for the sake of completeness, we briefly indicate how each of these term scan be bounded using the ones corresponding to $T_2$.
To begin with, consider the terms $T_{31}$ and $T_{32}$. One can perform identical analysis as for the term $T_{22}$ (see Appendix \ref{appx:proofofPropT_22}) and obtain the following bounds.
\begin{align}
    \EE[T_{31}] &\leq \text{exp}_2\left(\frac{3k}{n}\log{q}- \left( {3\log{q}} -H(V_1,V_2) -H(U) +  I(Y_1;U|X_1)_{\sigma_1}  + \delta_{T_{22}}\right) \right), \label{eq:boundsT_31} \\
    \EE[T_{32}] &\leq \text{exp}_2\left(\frac{2k+l}{n}\log{q}- \left( {3\log{q}} -H(V_1,V_2) -H(U) +  I(Y_1;U|X_1)_{\sigma_1}  + \delta_{T_{22}}\right) \right), \label{eq:boundsT_32}
\end{align}
where $\text{exp}_2(x) \deq 2^{x}.$
Note that the exponents in the right hand side terms \eqref{eq:boundsT_31} and \eqref{eq:boundsT_32} are always negative given the bound in Proposition \ref{prop:LemmaT_22} is true. Hence $T_{31}$ and $T_{32}$ can be made arbitrarily small for sufficiently large $n$ without any additional constraints.

Similarly, consider the terms $T_{33}$ and $T_{34}$. Using an identical analysis as for the term $T_{23}$ (see Appendix \ref{appx:proofofPropT_23}) we obtain
\begin{align}
    \EE[T_{33}] & \leq \text{exp}_2\left(R_1 + \frac{3k}{n}\log{q} - \left( {3\log{q}} -H(V_1,V_2) - H(U) +  I(Y_1;X_1,U)_{\sigma_1}  + \delta_{T_{23}}\right) \right), \label{eq:boundsT_33}\\
    \EE[T_{32}] &\leq \text{exp}_2\left(R_1 + \frac{2k+l}{n}\log{q} - \left( {3\log{q}} -H(V_1,V_2) - H(U) +  I(Y_1;X_1,U)_{\sigma_1}  + \delta_{T_{23}}\right) \right), \label{eq:boundsT_34}
\end{align}
Again observe that the exponents in the right hand side of \eqref{eq:boundsT_33} and \eqref{eq:boundsT_34} are always negative given the bound in Proposition \ref{prop:LemmaT_23} is true. Hence $T_{33}$ and $T_{34}$ can be made arbitrarily small for sufficiently large $n$ without any additional constraints.
Having completed the proof for the terms in $T$,  we now provide the result stating: NCC codes achieve capacity of a CQ-PTP channel (as discussed in the proof of Proposition \ref{prop:Lemma for S_23}).
\end{proof}

\section{Coset Codes for communicating over CQ-PTP}
\label{Sec:NCCAchieveCQ-PTPCapacity}
As discussed in Sec.~\ref{Sec:AchReg3to1CQIC}, here we shall build and analyze a NCC for a point-to-point CQ channel \cite{2013Bk_Wil} and employ it as a module for the $3to1$ CQ-IC result.
Towards that, we begin by formalizing the definition of a CQ-PTP code.
\begin{definition}
A CQ-PTP code $c_{m}=(n,\mathcal{I},e,\lambda_{\mathcal{I}})$ for a CQ-PTP $(\rho_{x} \in \mathcal{D}(\mathcal{H}_{Y}):x \in \mathcal{X})$ consists of (i) an index set $\mathcal{I}$, (ii) and encoder map $e:\mathcal{I}\rightarrow \mathcal{X}^{n}$ and a decoding POVM $\lambda_{\mathcal{I}} = \{ \lambda_{m} \in \mathcal{P}(\mathcal{H}_{Y}): m \in \mathcal{I} \}$. For $m \in \mathcal{I}$, we let $\rho_{c,m}^{\otimes n} = \otimes_{i=1}^{n}\rho_{x_{i}}$ where $e(m) = x_{1}\cdots x_{n}$.
\end{definition}

\begin{definition}
A CQ-PTP code $(n,\mathcal{I}=\mathcal{F}_{q}^{l},e,\lambda_{\mathcal{I}})$ is an NCC CQ-PTP if there exists an $(n,k,g_I,g_{O/I},b^n)$ NCC such that $e(m) \in \{u^{n}(a,m) : a \in \CalF_q^k\}$ for all $m \in \CalF_q^l$.
\end{definition}

\begin{theorem}
\label{Thm:NCCAchievesCQPTP}
 Given a CQ-PTP $(\rho_{v} \in \mathcal{D}(\mathcal{H}_{Y}): v \in \CalF_q)$ and a PMF $p_{V}$ on $\mathcal{V}$, $\epsilon > 0$ there exists a CQ-PTP code $c=(n,\mathcal{I}=\mathcal{F}_{q}^{l},e,\lambda_{\mathcal{I}})$ such that (i) $q^{-l}\sum_{\hatm \neq [\mathcal{I}]\setminus\{m\}}\tr(\lambda_{\hatm} \rho^{\otimes n}_{c,m}) \leq \epsilon$, (ii) $c=(n,\mathcal{I}=\mathcal{F}_{q}^{l},e,\lambda_{\mathcal{I}})$ is a NCC CQ-PTP, (iii) $\frac{k\log_{2}q}{n} >  \log_{2}q - H(V)$ and $\frac{(k+l)\log_{2}q}{n} < \log_2{q} - H(V) +\chi(\{p_v,\rho_v\}) $ for all $n$ sufficiently large.
\end{theorem}

\begin{proof}
The proof has two parts: (i) error probability analysis for a generic fixed code and (ii) an upper bound on the latter via code randomization.
\med\textit{Upper bound on Error Prob. for a generic fixed code} : Consider a generic NCC $(n,k,l,g_{I},g_{O/I},b^{n})$ with its range space $v^{n}(a,m) = ag_{I}\oplus_{q}mg_{O/I}\oplus_{q}b^{n}: (a,m) \in \mathcal{V}^{k} \times \mathcal{V}^{l}$. We shall use this and define a CQ-PTP code $(n,\mathcal{I}=\mathcal{F}_{q}^{l},e,\lambda_{\mathcal{I}})$ that is an NCC CQ-PTP. Towards that end, let $\theta(m) \define \sum_{a \in \mathcal{V}^{k}}\mathds{1}_{\left\{ v^{n}(a,m) \in T_{\delta}^{n}(p_{V}) \right\}}$ and 
\begin{eqnarray}
\label{Eqn:NoTypicalElements}
 s(m) \define \begin{cases} \{a \in \mathcal{V}^{K}: v^{n}(a,m) \in T_{\delta}^{n}(p_{V})\}&\mbox{if }\theta(m) \geq 1 \\\{0^{k}\}&\mbox{if }\theta(m) = 0,\end{cases}
 \nonumber
\end{eqnarray}
for each $m \in \mathcal{V}^{l}$. For $m \in \mathcal{V}^{l}$, a predetermined element $a_{m} \in s(m)$ is chosen. On receiving message $m \in \mathcal{V}^{l}$, the encoder prepares the  state $\rho_{m}^{\otimes n} \define \rho^{\otimes n}_{v^{n}(a_{m},m)} \define \otimes_{i=1}^{n}\rho_{v_{i}(a_{m},m)}$ and is communicated. The encoding map $e$ is therefore determined via the collection $(a_{m} \in s(m): m \in \mathcal{V}^{l})$.

Towards specifying the decoding POVM, 
for any $v^{n}\in \CalV^{n}$, let $\pi_{v^{n}}$ be the conditional typical projector as in \cite[Defn. 15.2.4]{2013Bk_Wil} with respect $\rho_v$ and
let $\pi_{\rho}$ be the (unconditional) typical projector of the state $\rho\define \sum_{v \in \mathcal{V}} p_{V}(v)\rho_{v}$ as in \cite[Defn. 15.1.3]{2013Bk_Wil}.
\footnote{We have done away with superscript $n$ - dimension of the underlying space - to reduce cluttter.} 
For $(a,m) \in \CalV^{k}\times \CalV^{l}$, we let $\pi_{a,m} \define \pi_{v^{n}(a,m)}\mathds{1}_{\{ v^{n}(a,m) \in T_{\delta}^{n}(p_{V}) \}}$. We let $\lambda_{\mathcal{I}} \define \{ \sum_{a \in \CalV^{k}}\lambda_{a,m} : m \in \mathcal{I}=\CalV^{l},\lambda_{-1} \}$, where  
\begin{align}
 \lambda_{a,m} \!\define\! \Big(\! \sum_{\hata \in \CalV^{k}}\! \sum_{\hatm \in \CalV^{l}} \!\!\gamma_{\hata,\hatm}\Big)^{-{1}/{2}}\!\!\! \gamma_{a,m}\Big(\! \sum_{\tildea \in \CalV^{k}} \sum_{\tildem \in \CalV^{l}}\!\! \gamma_{\tildea,\tildem}\Big)^{-{1}/{2}}\nonumber
\end{align}
$\lambda_{-1}\define I-\sum_{m \in \CalV^{l}}\sum_{a \in \CalV^{k}}\lambda_{a,m}$ and $\gamma_{a,m} \define \pi_{\rho}\pi_{a,m}\pi_{\rho}$. Since $0 \leq \gamma_{a,m} \leq I$, we have $0 \leq \lambda_{a,m} \leq I$. 
The latter lower bound implies $\lambda_{\mathcal{I}} \subseteq \mathcal{P}(\mathcal{H})$. The same lower bound coupled with the definition of the generalized inverse implies $0 \leq \sum_{a \in \CalV^{k}}\sum_{m \in \CalV^{l}}\lambda_{a,m} \leq I$. We thus have $0 \leq \lambda_{-1} \leq I$. 
It can be verified that $\lambda_{\mathcal{I}}$ is a POVM.
In essence, the elements of this POVM is identical to the standard POVMs \cite{2012Bk_Hol,2013Bk_Wil}, except the POVM elements corresponding to a coset have been added together. Indeed, since each coset corresponds to one message, there is no need to disambiguate within the coset.
We have thus associated an NCC $(n,k,l,g_{I},g_{O/I},b^{n})$ and a collection $(a_{m} \in s(m): m \in \mathcal{V}^{l})$ with a CQ-PTP code. 
The error probability of this code is
\begin{eqnarray}
 \label{Eqn:CQ-PTPErrorProbability}
 q^{-l}\sum_{m \in \mathcal{I}}\mbox{tr}((I-\sum_{a \in \CalV^{k}}\lambda_{a,m})\rho_{m}^{\otimes n}) \leq  q^{-l}\sum_{m \in \mathcal{I}}\mbox{tr}((I-\lambda_{a_{m},m})\rho_{m}^{\otimes n})
\nonumber
\end{eqnarray}
Denoting event $\mathscr{E}=\{ \theta(m)< 1 \}$, its complement $\mathscr{E}^{c}$ and the associated indicator functions $\mathds{1}_{\mathscr{E}}, \mathds{1}_{\mathscr{E}^c}$ respectively, a generic term in the RHS of the above sum satisfies
 \begin{align}
{\mbox{tr}((I-\lambda_{a_{m},m})\rho_{m}^{\otimes n})\mathds{1}_{\mathscr{E}^{c}}
 + \mbox{tr}((I-\lambda_{a_{m},m})\rho_{m}^{\otimes n})\mathds{1}_{\mathscr{E}}}  \leq& \mathds{1}_{\mathscr{E}^{c}}+\sum_{i=1}^{3}T_{2i},
 \end{align}
 where
 \begin{align}
 T_{21} = 2\mbox{tr}((I-\gamma_{a_{m},m})\rho_{m}^{\otimes n})\mathds{1}_{\mathscr{E}} , \quad
{T_{22}= 4\sum_{\hata \neq a_{m} }\mbox{tr}(\gamma_{\hata,m}\rho_{m}^{\otimes n})\mathds{1}_{\mathscr{E}},\quad  T_{23} = 4\sum_{\hatm \neq m}\sum_{\tildea  }\mbox{tr}(\gamma_{\tildea,\hatm}\rho_{m}^{\otimes n})\mathds{1}_{\mathscr{E}}},
 \nonumber
\end{align}
and the inequality follows by Hayashi-Nagaoka inequality \cite{hayashi2003general}. 
 for  $0 \leq S \leq I$, and $T \geq 0$,
with $S$ and $T$ identified as $ \gamma_{a_{m},m}$ and $ \sum_{\hata \neq a_{m}}  \gamma_{\hata,m} + \sum_{\hata \in \CalV^{k}} \sum_{\hatm \neq m} \gamma_{\hata,\hatm},$ respectively. 
Note that $S$ and $T$ satisfy the required hypothesis which can be verified from earlier stated facts.

\textit{Distribution of the Random Code} : The objects $g_{I}\in \mathcal{V}^{k \times n},g_{O/I} \in \mathcal{V}^{l \times n},b^{n} \in \mathcal{V}^{n}$ and the collection $(a_{m} \in s(m): m \in \mathcal{V}^{l})$ specify a NCC CQ-PTP code unambiguously. A distribution for a random code is therefore specified through a distribution of these objects. We let upper case letters denote the associated random objects, and obtain
\begin{eqnarray}
 \label{Eqn:}
 \CalP\!\left( \begin{array}{c} G_{I}=g_{I},G_{O/I}=g_{O/I}\\B^{n}=b^{n},A_{m}=a_{m}: m \in S(m) \end{array}  \right) = q^{-(k+l+1)n}\prod_{m \in \CalV^{l}}\frac{1}{\Theta(m)},\nonumber
\end{eqnarray}
and analyze the expectation of $\mathscr{E}$ and the terms $T_{2i}; i\in [1,3]$ in regards to the above random code. 
We begin by $\mathbb{E}_{\mathcal{P}}[\mathscr{E}] = \mathcal{P}(\sum_{a \in \CalV^{k}} \mathds{1}_{\{V^{n}(a,m) \in T_{\delta}^{n}(p_{V}) \}} < 1). $ For this, we provide the following proposition.  
\begin{prop}\label{prop:PTP:Lemma for T1}
There exist  $\epsilon_{T_1}(\delta), \delta_{T_1}(\delta),$ such that for  all  $\delta$ and sufficiently large $n$, we have $\EE_{\CalP}\left[\mathscr{E}\right] \leq \epsilon_{{T_1}}(\delta) $, if $ \frac{k}{n} \geq \log{q} - H(V) + \delta_S$, where  $\epsilon_{{S}},\delta_{S} \searrow 0$ as $\delta \searrow 0$.
\end{prop}
\begin{proof}
The proof follows from Appendix B of  \cite{MACwithStates_ArunPad_SandeepPra}. 
\end{proof}
We now consider $T_{21}$. Since this term can be bounded by a using straight-forward extension of the pinching technique described in \cite[Def. 15.2.4]{2013Bk_Wil}, we provide its complete details in \cite{arxiv_currentPaper}. 
\noindent We now analyze $\mathbb{E}_{\mathcal{P}}[T_{22}]$.
Denoting \begin{eqnarray}\label{def:JandK}
\mathcal{J} \define \left\{  \begin{array}{c} \Theta(m) \geq 1 ,\!V^{n}(\hat{a},\hatm=\hatx^{n}\\A_{m}=d,V^{n}(d,m)=x^{n} \end{array} \right\}\subseteq\! \mathcal{K} \define\! \left\{  \begin{array}{c} V^{n}(\hat{a},\hatm)=\hatx^{n}\\V^{n}(d,m)=x^{n} \end{array}\right\}
\end{eqnarray} we perform the following steps.
\begin{align}
 \mathbb{E}_{\mathcal{P}}[T_{22}] &= \sum_{\hata \in \CalV^{k}  }\mathbb{E}_{\mathcal{P}}[\mbox{tr}(\Gamma_{\hata,m}\rho_{m}^{\otimes n})\mathds{1}_{\{ \theta(m) \geq 1 \}}\mathds{1}_{\{\hata \neq A_{m}\}}] \nonumber \\
 &= \sum_{d \in \CalV^{k}}\sum_{\hata \in \CalV^{k} }\sum_{x^{n} \in T_{\delta}^{n}(p_{V})}\sum_{\hatx^{n} \in \CalV^{n}}\mathbb{E}\left[\mbox{tr}(\Gamma_{\hata,m}\rho_{m}^{\otimes n})\11_{\hata \neq d}\11_{\CalJ}\right]  \nonumber\\
&= \sum_{d \in \CalV^{k}}\sum_{\hata \neq d }\sum_{x^{n} \in T_{\delta}^{n}(p_{V})}\sum_{\hatx^{n} \in \CalV^{n}}\mathbb{E}\left[\mbox{tr}(\Gamma_{\hata,m}\rho_{m}^{\otimes n})\11_{\CalJ}\right] \nonumber
\end{align}
where the restriction of the summation $x^{n}$ to $T_{\delta}^{n}(p_{V})$ is valid since $S(m) \geq \tau_{c} > 1$ forces the choice $A_{m} \in S(m)$ such that $V^{n}(A_{m},m) \in T_{\delta}^{n}(p_{V})$. Going further, we have 
\begin{align}\label{eq:PTP_T22}
 \mathbb{E}_{\mathcal{P}}[T_{22}] &= \sum_{\substack{d,\hata \in \CalV^{k}\\ \hata \neq d}}\sum_{x^{n} \in T_{\delta}^{n}(p_{V})}\sum_{\hatx^{n} \in T_{\delta}^{n}(p_{V})}\mathbb{E}\left[\mbox{tr}(\pi_{\rho}\pi_{\hatx^{n}}\pi_{\rho}\rho_{x^{n}}^{\otimes n})\11_{\CalJ}\right]\nonumber\\
 &= \sum_{\substack{d,\hata \in \CalV^{k}\\ \hata \neq d}}\sum_{x^{n} \in T_{\delta}^{n}(p_{V})}\sum_{\hatx^{n} \in T_{\delta}^{n}(p_{V})}\mbox{tr}(\pi_{\rho}\pi_{\hatx^{n}}\pi_{\rho}\rho_{x^{n}}^{\otimes n})\mathcal{P}(\CalJ)\nonumber\\
 &\stackrel{(a)}{\leq}\sum_{\substack{d,\hata \in \CalV^{k}\\ \hata \neq d}}\sum_{\hat{x}^n  \in T_{\delta}^{n}(p_{V})}\mbox{tr}(\pi_{\hatx^{n}}\pi_{\rho} )\mathcal{P}(\CalJ)2^{-n\left[S(\rho)-H(p_{V})+\epsilon_V\right]} \nonumber\\
&\stackrel{(b)}{\leq} \sum_{\substack{d,\hata \in \CalV^{k}\\ \hata \neq d}}\sum_{{\hatx^{n} \in T_{\delta}^{n}(p_{V})}}\mbox{tr}(\pi_{\hatx^{n}}\pi_{\rho})\mathcal{P}(\mathcal{K})2^{-n\left[S(\rho)-H(p_{V})+\epsilon_V\right]} \nonumber\\
 &\stackrel{(c)}{=} \sum_{\substack{d,\hata \in \CalV^{k}\\ \hata \neq d}}\sum_{{\hatx^{n} \in T_{\delta}^{n}(p_{V})}}\mbox{tr}(\pi_{\hatx^{n}}\pi_{\rho})\frac{1}{q^{2n}}2^{-n\left[S(\rho)-H(p_{V})+\epsilon_V\right]}\nonumber\\
 &\stackrel{(d)}{\leq} 2^{-n\left[ \chi(\{p_{V};\rho_{v}\})+ \epsilon_V  -2H(p_{V}) -\frac{2k}{n}\log{q} + 2\log{q} \right]} 
\end{align}
where the restriction of the summation $\hatx^{n}$ to $T_{\delta}^{n}(p_{V})$ follows from the fact that $\pi_{\hat{x}^{n}}$ is the zero projector if $\hatx^{n} \notin T_{\delta}^{n}(p_{V})$, (a) follows from the operator inequality $ \sum_{x^n \in T_{\delta}(p_{V})}\pi_{\rho}\rho_{x^{n}}\pi_{\rho} \leq 2^{n(H(p_{V})+\epsilon_V(\delta)) }\pi_{\rho}\rho^{\otimes n}\pi_{\rho} \leq 2^{n(H(p_{V})+\epsilon_V(\delta)-S(\rho)) }\pi_{\rho}$ found in \cite[Eqn. 20.34, 15.20]{2017BkOnline_Wil}, (b) follows from Definition \ref{def:JandK}, (c) follows from pairwise independence of the distinct codewords, and (d) follows from $\pi_{\rho} \leq I $ and \cite[Eqn. 15.77]{2017BkOnline_Wil} and $\epsilon_V(\delta) \searrow 0$ as $\delta \searrow 0$.
 We now derive an upper bound on $\mathbb{E}_{\mathcal{P}}[T_{23}]$. 
We have
\begin{align}
\mathbb{E}&_{\mathcal{P}}[T_{23}] = \sum_{d,\hata \in \CalV^{k}}\sum_{\hatm \neq m }\sum_{\substack{x^{n},\hatx^{n} \in \\T_{\delta}^{n}(p_{V})}}\!\!\!\mathbb{E}\!\left[\mbox{tr}(\pi_{\rho}\Pi_{\hata,\hatm}\pi_{\rho}\rho_{A_{m},m}^{\otimes n})\mathds{1}_{\mathcal{J}}\right] \nonumber\\
&{=} \sum_{d,\hata \in \CalV^{k}}\sum_{\hatm \neq m }\sum_{\substack{x^{n},\hatx^{n} \in T_{\delta}^{n}(p_{V})}}\mbox{tr}(\pi_{\hatx^{n}}\pi_{\rho}\rho_{x^{n}}^{\otimes n}\pi_{\rho})\mathcal{P}(\mathcal{J})\nonumber\\
 &{\leq} \sum_{d,\hata \in \CalV^{k}}\sum_{\hatm \neq m }\sum_{\substack{\hatx^{n} \in T_{\delta}^{n}(p_{V})}}\!\!\!\!\mbox{tr}(\pi_{\hatx^{n}}\pi_{\rho})\mathcal{P}(\mathcal{J})2^{-n\left[S(\rho)-H(p_{V})+\epsilon_V\right]}\nonumber\\
 &{\leq}  \sum_{d,\hata \in \CalV^{k}}\sum_{\hatm \neq m }\sum_{\substack{\hatx^{n} \in T_{\delta}^{n}(p_{V})}}\!\!\!\!\mbox{tr}(\pi_{\hatx^{n}}\pi_{\rho})\mathcal{P}(\mathcal{K})2^{-n\left[S(\rho)-H(p_{V})+\epsilon_V\right]}\nonumber\\
&{=} \sum_{d,\hata \in \CalV^{k}}\sum_{\hatm \neq m }\sum_{\substack{\hatx^{n} \in T_{\delta}^{n}(p_{V})}}\!\!\!\!\mbox{tr}(\pi_{\hatx^{n}}\pi_{\rho})\frac{1}{q^{2n}}2^{-n\left[S(\rho)-H(p_{V})+\epsilon_V\right]}\nonumber\\
 &{\leq} \;\; 2^{-n\left[ \chi(\{p_{V};\rho_{v}\})+ 2\log_{2}q-2H(p_{V}) -\frac{2k+l}{n}\log_{2}q +\epsilon_V \right]} \nonumber
\end{align}
where (a) follows from the operator inequality $$ \sum_{x^n \in T_{\delta}(p_{V})}\pi_{\rho}\rho_{x^{n}}\pi_{\rho} \leq 2^{n(H(p_{V})+\epsilon_V(\delta)) }\pi_{\rho}\rho^{\otimes n}\pi_{\rho} \leq 2^{n(H(p_{V})+\epsilon_V(\delta)-S(\rho)) }\pi_{\rho}$$ found in \cite[Eqn. 20.34, 15.20]{2017BkOnline_Wil}, (b) follows from Definition \ref{def:JandK}, (c) follows from pairwise independence of the distinct codewords and (d) follows from $\pi_{\rho} \leq I $ and \cite[Eqn. 15.77]{2017BkOnline_Wil} and $ \epsilon_V(\delta) \searrow 0$ as $\delta \searrow 0$.
where the inequalities above uses similar reasoning as in \eqref{eq:PTP_T22}.
We have therefore obtained three bounds 
$\frac{k}{n} > 1-\frac{H(p_{V})}{\log_{2} q} $, $\frac{2k}{n} < 2+\frac{\chi(\{p_{V};\rho_{v}\})-2H(p_{V})}{\log_{2} q} $, $\frac{2k+l}{n} < 2+\frac{\chi(\{p_{V};\rho_{v}\})-2H(p_{V})}{\log_{2} q} $. A rate of $\chi(\{p_{V};\rho_{v}\})-\epsilon$ is achievable by choosing $\frac{k}{n} = 1-\frac{H(p_{V})}{\log_{2} q}+\frac{\epsilon}{2}$, $\frac{l}{n} = \frac{\chi(\{p_{V};\rho_{v}\})-\epsilon\log_2\sqrt{q}}{\log_{2} q}$ thus completing the proof.
\end{proof}

\section{Rate-region using NCC and message splitting for $3$to$1-$ CQIC}
 
 \begin{theorem} \label{thm:messageSplitThm}
Given a $3to1$-CQIC $(\rho_{\ulinex} \in \mathcal{D}(\mathcal{H}_{\ulineY}): \ulinex \in \ulineCalX)$ and a PMF $p_{U_2U_3V_2V_3X_2X_3} = p_{U_2V_2X_2}p_{U_3V_3X_3}$ on $\mathcal{U}_1\cross\mathcal{V}_1\cross\mathcal{X}_1\cross\mathcal{U}_2\cross\mathcal{V}_2\cross\mathcal{X}_2$ where $\mathcal{V}_1 = \mathcal{V}_2 = \mathcal{F}_q$, a rate triple is achievable if it satisfies the following: $R_j  \leq I(U_jX_j;Y_j)_{\sigma_j},$
\begin{align}
     R_1 \!&\leq\! \min_{j=2,3}\{0, H(U_j) - H(W|Y_1)_{\sigma_1} \} 
    + I(X_1;WY_1)_{\sigma_1}\nonumber \\
    R_1 \!+\! R_j \!&\leq\! I(X_j;Y_j|U_j)_{\sigma_j} \!\!\!+\! I(X_1;W,Y_1)_{\sigma_1}\!\!\!  + \!H(U_j)\! -\! H(W|Y_1)_{\sigma_1} \nonumber
\end{align}
for $j =2,3$, where
\begin{align*}
    \sigma_1^{\ulineY}&\deq  \sum_{x_1\in\mathcal{X}_1,w \in \CalF_q}\!\!\!p_{X_1}(x_1)p_{W}(w)\rho_{x_1,w}^{\ulineY}\otimes \ketbra{x_1}\otimes \ketbra{w},\nonumber \\
    \rho_{x_1,w}^{\ulineY} &\deq \!\!\sum_{\substack{u_2,v_2,x_2\\u_3,v_3,x_3}}\!\!p_{V_2,V_3U_2U_3X_2X_3|W}(v_2,v_3,u_2,u_3,x_2,x_3|w)\rho_{\ulinex}^{\ulineY} \nonumber \\
    \sigma_2 &= \!\!\!\!\!\sum_{v_1,v_2,v_3}\!\!\!\!\! p_{U_2U_3V_2V_3\ulineX}(u_2,u_3,v_2,v_3,\ulinex)\rho_{\ulinex}^{\ulineY}\otimes_{j=2}^{3} \ketbra{u_j,x_j}
\end{align*}for $W \deq U_2 \oplus U_3,$ and $\{\ket{u_j}\}$ and $\{\ket{x_j}\}$ as some orthonormal basis on $\CalH_Y$ for $j={2,3}.$
\end{theorem}
\begin{proof}
In view of the detailed proof provided for Thms.~\ref{thm:3to1CQIC}, \ref{Thm:NCCAchievesCQPTP}, we only provide an outline. A complete proof of this theorem is reserved for an enlarged version of this preprint.

In the coding scheme of Thm.~\ref{thm:3to1CQIC}, Rx 1 decodes a bivariate function of Tx 2 and Tx 3's inputs. In general, decoding just a bivariate function of Tx 2 and Tx 3's inputs is insufficient. It is necessary for the coding scheme to permit Rx 1 decode univariate functions of the Tx 2 and Tx 3's inputs as well. Therefore an enhanced coding scheme, will split Tx 2 and Tx 3's transmissions into two parts respectively. For $j =2,3$, let $U_{j},V_{j}$ denote the splitting of Tx $j$'s input. Here, $U_{2},U_{3} \in \mathcal{F}_{q}$ take values in a common finite field. $U_{2}$ and $U_{3}$ are communicated via a common nested coset code. $V_{2},V_{3},X_{1}$ are built via conventional unstructured codes. Since this is a $3$to$1-$IC, Tx $1$ does not split its input $X_{1}$. 

Observe that, for $j \in 2,3$, Rx $j$ has to decode a $U_{j}, V_{j}$, one component of which is encoded via a nested coset code, and the other component which is encoded via a conventional unstructured code. The analysis of its decoding is similar to the analysis of Tx $1$'s decoding in proof of Theorem \ref{thm:3to1CQIC}. Indeed, in proof of Theorem \ref{thm:3to1CQIC}, Tx $1$ decoded from its unstructured code and the bivariate component of the interference that was encoded via a nested coset code. This provides the outline for the analysis of Rx 2 and 3. Rx $1$ has to decode 4 components - 1 structured ($U_{2}\oplus U_{3}$) and $3$ unstructured $V_{2}, V_{3}, X_{1}$. We adopt successive-simulataneous decoding wherein two code words are decoded at each stage of a 2 stage process. This outline does not substitute a complete and detailed proof which is will be provided in an enlarged version of this preprint.
\end{proof}
By choosing $W=\phi$, we can recover the $\mathscr{U}\mathcal{S}\mathcal{B}-$rate region from the above inner bound.

\section{Proof of Lemmas}
\subsection{Proof of Lemma \ref{lem:LemmaPinching}}\label{appx:proofLemmaPinching}
We provide the following generalized version proof of the pinching result, which can be used to prove the inequalities stated in the lemma.
\begin{lemma}
 \label{Lem:CharHighProbSubAKAPinching}
 Suppose (i) $\mathcal{A},\mathcal{B}$ are finite sets, (ii) $p_{AB}$ is a PMF on $\mathcal{A} \times \mathcal{B}$, (iii) $(\rho_{b} \in \mathcal{D}(\mathcal{H}): b \in \mathcal{B})$ is a collection of density operators, $\rho_{a} \define \sum_{b \in \mathcal{B}}p_{B|A}(b|a)\rho_{b}$ for $a \in \mathcal{A}$ and $\rho = \displaystyle\sum_{a \in \mathcal{A}}p_{A}(a)\rho_{a} = \displaystyle\sum_{b \in \mathcal{B}} p_{B}(b)\rho_{b}$. There exists a strictly positive $\mu>0$, whose value depends only on $p_{AB}$, such that for every $\delta > 0$, there exists a $N(\delta) \in \naturals$ such that for all $n  \geq N(\delta)$, we have
 \begin{eqnarray}
  \label{Eqn:CharHighProbSubAKAPinching-1}
  \tr(\Pi_{\rho}^{\delta}\Pi_{a^{n}}^{\delta}\Pi_{\rho}^{\delta}\rho_{b^{n}}) \geq 1-\exp\{ -n\lambda\delta^{2} \}
  \nonumber
 \end{eqnarray}
whenever $(a^{n},b^{n}) \in T_{\frac{\delta}{4}}^{n}(p_{AB})$ where $\Pi_{a^{n}}^{\delta}$ is the conditional typical projector of $\rho_{a^{n}}= \otimes_{t=1}^{n}\rho_{a_{t}}$ \cite[Defn.~15.2.4]{2013Bk_Wil} and $\Pi_{\rho}^{\delta}$ is the unconditional typical projector \cite[Defn.~15.1.3]{2013Bk_Wil} of $\rho^{\otimes n}$ .
\end{lemma}

\noindent \textbf{Proof:}  We rename $\mathcal{A}=\mathcal{V}$, $\mathcal{B}=\mathcal{X}$, $p_{AB}=p_{VX}$, $a$ as $v$ and $b$ as $x$.
 This gives,
\begin{align}
\tr(\Pi_{\rho}^{\delta}
\Pi_{v^n}^{\delta} \Pi_{\rho}^{\delta} \rho_{x^n}) &=
\tr(\Pi_{\rho}^{\delta}
\Pi_{v^n}^{\delta} \rho_{x^n} \Pi_{\rho}^{\delta})\nonumber \\
&\geq \tr(\Pi_{v^n}^{\delta} \rho_{x^n}) -\frac{1}{2} \left\|
\rho_{x^n} -\Pi_{\rho}^{\delta} \rho_{x^n} \Pi_{\rho}^{\delta} \right\|.
\end{align}

In the following we derive a lower bound on $\tr(\Pi_{v^n}^{\delta}\rho_{x^n})$ and
derive an upper bound on $\left\| \rho_{x^n} -\Pi_{\rho}^{\delta} \rho_{x^n} \Pi_{\rho}^{\delta}
\right\|$. Toward the deriving the former, we recall that we have
$(v^n,x^n) \in T_{\delta/2}{^n}(p_{VX})$. Let us define:
\begin{equation}
  p_{Y|XV}(y|x,v) :=\braket{e_{y|v}|\rho_x|e_{y|v}},
  \label{eq:eq(2)}
  \end{equation}
for all $(x,v,y) \in \mathcal{X} \times \mathcal{V} \times \mathcal{Y}$.

Clearly, we have $p_{Y|XV}(y|x,v) \geq 0$, and
$\sum_{y \in \mathcal{Y}} p_{Y|XV}(y|x,v)=\sum_{y \in \mathcal{Y}} \braket{e_{y|v} | \rho_x|
  e_{y|v}}=\tr(\rho_x)=1$. Hence we see that $p_{Y|XV}$ is a stochastic matrix.

Next we note that
\begin{align}
  \sum_{x \in \mathcal{X}} p_{Y|XV}(y|x,v)p_{XV}(x,v) &=
  \sum_{x \in \mathcal{X}} p_{XV}(x,v) \braket{e_{y|v} |\rho_x | e_{y|v}} \nonumber \\
  &= p_{V}(v) \braket{e_{y|v} |\sum_{x \in \mathcal{X}} p_{X|V}(x|v) \rho_x |e_{y|v}} \nonumber \\
  &= p_V(v)\braket{e_{y|v}|\rho_v |e_{y|v}} = p_V(v) q_{Y|V}(y|v),
  \label{equation_1}
  \end{align}
   where we have used the spectral decomposition of $\rho_v$. 

   Observe that  if $(x^n,v^n) \in T_{\delta/4}^n(p_{XV})$, and
   $y^n\in T_{\delta}^n(p_{XV}p_{Y|XV}|x^n,v^n)$, then we have
   $(x^n,v^n,y^n) \in T_{\delta}^n(p_{XV}p_{Y|XV})$. This implies that we have
   $(v^n,y^n) \in T_{\delta}^n(p_{VY})$, where $p_{VY}$ is the marginal
   of $p_{XV}p_{Y|XV}$. Using this and (\ref{equation_1}), we see that
   $(v^n,y^n) \in T_{\delta}^n(p_V q_{Y|V})$. In summary, we see that if
   $(x^n,v^n) \in T_{\delta/4}(p_{XV})$, then we have
   \[
   T_{\delta}^n(p_{XV}p_{Y|XV}|x^n,v^n) \subseteq \left\{y^n: (v^n,y^n) \in
   T_{\delta}^n(p_V q_{Y|V})  \right\}.
   \]
  
   We are now set to provide the promised lower bound. Consider
   \begin{align}
     \tr(\Pi_{v^n}\rho_{x^n}) &= \tr \left(\left[ \sum_{y^n: (v^n,y^n) \in
         T_{\delta}^n(p_Vq_{Y|V})} \bigotimes_{t=1}^n \ket{e_{y_t|v_t}} \bra{e_{y_t|v_t}} \right]
   \left[   \bigotimes_{j=1}^n \rho_{x_j} \right] \right) \\
     &= \tr \left(\left[ \sum_{y^n: (v^n,y^n) \in
       T_{\delta}^n(p_Vq_{Y|V})} \bigotimes_{t=1}^n \ket{e_{y_t|v_t}} \bra{e_{y_t|v_t}}
     \rho_{x_t} \right] \right) \\
   &=\sum_{y^n: (v^n,y^n) \in
     T_{\delta}^n(p_Vq_{Y|V})} \prod_{t=1}^n \braket{e_{y_t|v_t}|\rho_{x_t}|e_{y_t|v_t}} \\
   &\geq \sum_{y^n \in
     T_{\delta}^n(p_{XV}p_{Y|XV}|x^n,v^n))} \prod_{t=1}^n p_{Y|XV}(y_t|x_t,v_t) \\
   &\geq 1- 2|\mathcal{X}||\mathcal{Y}||\mathcal{V}| \exp \left\{-\frac{2n \delta^2
   p_{XVY}(x^*,v^*,y^*)}{4(\log(|\mathcal{X}||\mathcal{Y}||\mathcal{V}|))^2} \right\},
   \end{align}
   where we used the definition (\ref{eq:eq(2)}) in the last equality.
   
   We next provide the upper bound. Note from the Gentle measurements lemma
   \cite[Lemma 9.4.2]{2013Bk_Wil}, we have $\| \rho_{x^n} -\Pi_{\rho}^{\delta}\rho_{x^n}
   \Pi_{\rho}^{\delta}|| \leq 3\sqrt{\epsilon}$ if $\tr(\Pi_{\rho}^{\delta}
   \rho_{x^n})\geq 1-\epsilon$. In the following we provide a lower bound on
  $\tr(\Pi_{\rho}^{\delta}   \rho_{x^n})$. 
   Recall that $\Pi_{\rho}^{\delta} =\sum_{y^n \in T_{\delta}^n(s_Y)} \bigotimes_{t=1}^n
     \ket{g_{y_t}}\bra{g_{y_t}}$, where
     \[
\rho=\sum_{y \in \mathcal{Y}} s_Y(y) \ket{g_y}\bra{g_y},
\]
is the spectral decomposition of $\rho$, and $\rho=\sum_{x \in \mathcal{X}} p_X(x)\rho_x$.
Let $\hat{p}_{Y|X}(y|x):=\braket{g_y|\rho_x|g_y}$, for all
$(x,y)\in \mathcal{X} \times \mathcal{Y}$. Note that $\hat{p}_{Y|X}$ is not
related to $p_{Y|X}$ defined previously. We note that $\hat{p}_{Y|X}(y|x)\geq 0$, and
$\sum_{y \in \mathcal{Y}} \hat{p}_{Y|X}(y|x)=\sum_{y \in \mathcal{Y}} \braket{g_y|\rho_x|g_y}=
\tr(\rho_x)=1$ for all $x \in \mathcal{X}$.
Thus we see that $\hat{p}_{Y|X}$ is a stochastic matrix. It can also be noted that
\[
\sum_{x\in \mathcal{X}} \hat{p}_{Y|X}(y|x)p_X(x) =\braket{g_y|\sum_{x \in \mathcal{X}}
  p_X(x)\rho_x|g_y}=\braket{g_y|\rho|g_y}=s_Y(y),
\]
for all $y \in \mathcal{Y}$. 
This implies that the condition  $y^n \in T_{\delta}^n(s_Y)$ is
equivalent to the condition $y^n \in T_{\delta}^n(\hat{p_{Y}})$,
where $\hat{p_Y}(y)=\sum_{x \in \mathcal{X}} \hat{p}_{Y|X}(y|x)p_X(x)$. Moreover,
if $x^n \in T_{\delta/2}^n(p_X)$, and
$y^n \in T_{\delta}^n(p_X \hat{p}_{Y|X}|x^n)$, then we have
$(x^n,y^n) \in T_{\delta}^n(p_X\hat{p}_{Y|X})$. Consequently, we have
$y^n\in T_{\delta}^n(\hat{p_Y})$, which in turn implies that
$y^n \in T_{\delta}^n(s_Y)$. In essense, we have that if $x^n \in T_{\delta/2}^n(p_X)$
then $T_{\delta}^n(p_X\hat{p}_{Y|X}|x^n) \subseteq T_{\delta}^n(s_Y)$.
Now we are set to provide the lower bound on $\tr(\Pi_{\rho}^{\delta}\rho_{x^n})$ as follows:
\begin{align}
  \tr(\Pi_{\rho}^{\delta} \rho_{x^n}) &= \tr \left(\sum_{y^n \in T_{\delta}(s_Y)} \bigotimes_{t=1}^n
    \ket{g_{y_t}}\bra{g_{y_t}} \rho_{x_t} \right)=\sum_{y^n \in T_{\delta}(s_Y)}
      \prod_{t=1}^n \braket{g_{y_t} |\rho_{x_t} |g_{y_t}} \\
      &= \sum_{y^n \in T_{\delta}(s_Y)}
      \prod_{t=1}^n \hat{p}_{Y|X}(y_t|x_t) \geq \sum_{y^n \in T_{\delta}(\hat{p}_{Y|X}p_X|x^n)} 
        \prod_{t=1}^n \hat{p}_{Y|X}(y_t|x_t) \\
         &\geq 1-2|\mathcal{X}||\mathcal{Y}| \exp \left\{ -\frac{2n \delta^2 p_X^2(x^*)
          \hat{p}_{Y|X}^2(y|x)}{4(\log(|\mathcal{X}||\mathcal{Y}|))^2} \right\}. 
  \end{align}
We therefore have
\[
\| \rho_{x^n} -\Pi_{\rho}^{\delta} \rho_{x^n} \Pi_{\rho}^{\delta}\| \leq
6 |\mathcal{X}||\mathcal{Y}| \exp \left\{ -\frac{2n \delta^2 p_X^2(x^*)
  \hat{p}_{Y|X}^2(y|x)}{4(\log(|\mathcal{X}||\mathcal{Y}|))^2} \right\},
\]
and
\[
\tr (\Pi_{v^n} \rho_{x^n} ) \geq 1-2|\mathcal{X}||\mathcal{Y}|||\mathcal{V}|
\frac{2n \delta^2 p_X^2(x^*)
          \hat{p}_{Y|X}^2(y|x)}{4(\log(|\mathcal{X}||\mathcal{Y}|))^2},
\]
thereby permitting us to conclude that
\[
\tr (\Pi_{\rho}^{\delta} \Pi_{v^n}^{\delta} \Pi_{\rho}^{\delta} \rho_{x^n}) \geq
\tr (\Pi_{v^n}^{\delta} \rho_{x^n} )-\frac{1}{2} \| \rho_{x^n}-\Pi_{\rho}^{\delta} \rho_{x^n}
\Pi_{\rho}^{\delta} \| \geq 1-\frac{2n \delta^2 p_X^2(x^*)
          \hat{p}_{Y|X}^2(y|x)}{4(\log(|\mathcal{X}||\mathcal{Y}|))^2},
\]
if $(x^n,v^n) \in T_{\delta/2}^n(p_{XV})$. 

\section{Proof of Propositions}
\subsection{Proof of Proposition \ref{prop:LemmaT_21} }\label{appx:proofofPropT_21}
We begin by defining the sets $\CalJ$ and $\CalK$ as 
\begin{align}
    \CalJ \deq \left\{\begin{array}{c}V^n_2(a_2,m_2) = v_2^n,\alpha_2(m_2) = a_2, \Theta_1(m_1) > 0\\ V^n_3(a_3,m_3) = v_3^n, \alpha_3(m_3) = a_3,\Theta_2(m_2) > 0\end{array}\right\}
  \subseteq \CalK \deq \{V^n_2(a_2,m_2) = v_2^n, V^n_3(a_3,m_3) = v_3^n\}
\end{align}
Now we simplify $\rho_{c,\ulinem}^{Y_1} $ as
\begin{align}\label{eq:SimplifyRho}
    \rho_{c,\ulinem}^{Y_1} &= \sum_{v_2^n, v_3^n \in \CalF_q^n}\sum_{x_1^n\in\CalX_1^n}\sum_{x_2^n,x_3^n \in \CalX^n_2\otimes\CalX_2^n}\!\!\!\!\!\!\! p_{X_2|V_2}(x^n_2|v^n_2)p^n_{X_3|V_3}(x^n_3|v^n_3)\rho_{x_1^nx_2^nx_3^n}^{Y_1}\11_{\{x_1^n(m_1) = x_1^n,v_2^n(\alpha_2(m_2),m_2) = v^n_2,v_3^n(\alpha_3(m_3),m_3) = v^n_3\}} \nonumber \\ 
    & = \sum_{v_2^n, v_3^n \in \CalF_q^n}\sum_{x_1^n \in \CalX_1^n}\rho_{x_1^nv_2^nv_3^n}^{Y_1}\11_{\{x_1^n(m_1) = x_1^n,v_2^n(\alpha_2(m_2),m_2) = v^n_2,v_3^n(\alpha_3(m_3),m_3) = v^n_3\}} \nonumber \\ 
    & = \sum_{x_1^n \in \CalX_1^n}\sum_{v_2^n, v_3^n \in \CalF_q^n}\sum_{a_2,a_3 \in \CalF_q^k} \rho_{x_1^nv_2^nv_3^n}^{Y_1} \11_{\{x_1^n(m_1) = x_1^n\}}\11_{\CalJ},
\end{align}
where the above two equalities are based on the encoding rules employed by the encoders and the last one follows from the definition of  $\CalJ$.
Using the above simplification in the term $T_{21}$, we get
\begin{align}
    T_{21}(\ulinem) & = \sum_{m_1'\neq m_1}\sum_{x_1^n \in \CalX_1^n}\sum_{v_2^n, v_3^n \in \CalF_q^n}\sum_{a_2,a_3 \in \CalF_q^k} \tr(\gamma_{m_1'}^{a,l} \pial\rho_{x_1^nv_2^nv_3^n}^{Y_1} \pial )\11_{\{x_1^n(m_1) = x_1^n\}}\11_{\CalJ}, \nonumber \\
    & = \sum_{m_1'\neq m_1}\sum_{\substack{x_1^n, \hat{x}_1^n \in \CalX_1^n\\}}\sum_{v_2^n, v_3^n \in \CalF_q^n}\sum_{a_2,a_3 \in \CalF_q^k} \sum_{u^n \in \CalF_q^n} \tr(\gammaXhatU \pi_{u^n}\rho_{x_1^nv_2^nv_3^n}^{Y_1} \pi_{u^n} )\11_{\{x_1^n(m_1) = x_1^n, x_1^n(m'_1) = \hat{x}_1^n\}}\11_{\{u^n = v_2^n\oplus v_3^n\}}\11_{\CalJ}, \nonumber
\end{align}
where $\gammaXhatU$ is defined as $\gammaXhatU \deq \pi_{\rho}\pi_{\hat{x}_1^n}\pi_{\hat{x}_1^n,u^n}\pi_{\hat{x}_1^n}\pi_{\rho} $
Taking expectation of the above term, we obtain
\begin{align}\label{eq:T21Ana}
    \EE[T_{21}] & = \sum_{m_1'\neq m_1}\sum_{\substack{x_1^n, \hat{x}_1^n \in \CalX_1^n}}\sum_{\substack{v_2^n, v_3^n \in \CalF_q^n \\ a_2,a_3 \in \CalF_q^k}} \sum_{u^n \in \CalF_q^n} \tr(\gammaXhatU \pi_{u^n}\rho_{x_1^nv_2^nv_3^n}^{Y_1} \pi_{u^n} )\PP(x_1^n(m_1) = x_1^n, x_1^n(m'_1) = \hat{x}_1^n)\11_{\{u^n = v_2^n\oplus v_3^n\}}\PP(\CalJ)\nonumber \\
    & \labelrel\leq{eq:T21label1} \sum_{m_1'\neq m_1}\sum_{\substack{x_1^n, \hat{x}_1^n \in \CalX_1^n}}\sum_{\substack{ a_2,a_3 \in \CalF_q^k}} \sum_{\substack{v_2^n\in\TDelta(V_2),\\ v_3^n \in \TDelta(V_3)}} \sum_{u^n \in \TDelta(U)} \tr(\gammaXhatU \pi_{u^n}\rho_{x_1^nv_2^nv_3^n}^{Y_1} \pi_{u^n} )p_{X_1}^n(x_1^n)p_{X_1}^n(\hat{x}_1^n)\11_{\{u^n = v_2^n\oplus v_3^n\}}\PP(\CalK) \nonumber\\
    & \labelrel\leq{eq:T21label2} \frac{2^{nR_1}q^{2k}}{q^{2n}} \sum_{\substack{x_1^n, \hat{x}_1^n \in \CalX_1^n}}   \sum_{u^n \in \TDelta(U)} \tr(\gammaXhatU \pi_{u^n} \left(\sum_{\substack{v_2^n\in\TDelta(V_2),\\ v_3^n \in \TDelta(V_3)}}\rho_{x_1^nv_2^nv_3^n}^{Y_1}\right) \pi_{u^n} )p_{X_1}^n(x_1^n)p_{X_1}^n(\hat{x}_1^n)\11_{\{u^n = v_2^n\oplus v_3^n\}} \nonumber\\
    & \labelrel\leq{eq:T21label3} \frac{2^{nR_1}q^{2k}}{q^{2n}}2^{n(H(V_2,V_3|U) + \delta_{u_2})} \sum_{{\hat{x}_1^n \in \CalX_1^n}}  p_{X_1}^n(\hat{x}_1^n) \sum_{u^n \in \TDelta(U)} \tr(\gammaXhatU \pi_{u^n}  \left(\sum_{{{x}_1^n \in \CalX_1^n}} p_{X_1}^n(x_1^n) \rho_{x_1u^n}^{Y_1}\right) \pi_{u^n} ) \nonumber\\
     & \labelrel={eq:T21label4} \frac{2^{nR_1}q^{2k}}{q^{2n}}2^{n(H(V_2,V_3|U) + \delta_{u_2})} \sum_{{\hat{x}_1^n \in \CalX_1^n}} p_{X_1}^n(\hat{x}_1^n)   \sum_{u^n \in \TDelta(U)} \tr(\gammaXhatU \pi_{u^n}\rho_{u^n}^{Y_1}\pi_{u^n} )\nonumber\\
     & \labelrel\leq{eq:T21label5} \frac{2^{nR_1}q^{2k}}{q^{2n}}2^{n(H(V_2,V_3|U) + \delta_{u_2})}2^{-n(S(Y_1|U)_{\sigma_1} - \delta_{u_3})} \sum_{{\hat{x}_1^n \in \CalX_1^n}} p_{X_1}^n(\hat{x}_1^n)   \sum_{u^n \in \TDelta(U)} \tr(\pi_{\rho}\pi_{\hat{x}_1^n}\pi_{\hat{x}_1^n,u^n}\pi_{\hat{x}_1^n}\pi_{\rho} \pi_{u^n} )\nonumber \\
     & \labelrel\leq{eq:T21label6} \frac{2^{nR_1}q^{2k}}{q^{2n}}2^{n(H(V_2,V_3|U) + \delta_{u_2})}2^{-n(S(Y_1|U)_{\sigma_1} - \delta_{u_3})} \sum_{{\hat{x}_1^n \in \CalX_1^n}} p_{X_1}^n(\hat{x}_1^n)   \sum_{u^n \in \TDelta(U)} \tr(\pi_{\hat{x}_1^n,u^n}),
\end{align}
where \eqref{eq:T21label1} follows from the presence of indicators $\theta_1(m_1)>0,$ $\theta_2(m_2) > 0$, and definition of $\CalK$, \eqref{eq:T21label2} follows by observing that $\PP(\CalK) = \frac{1}{q^{n}q^{n}}$, and \eqref{eq:T21label3} follows by using the following arguments for any $u^n \in \TDelta(U),$
\begin{align}\label{eq:simplifyRhoV1V2}
\sum_{\substack{v_2^n\in\TDelta(V_2),\\ v_3^n \in \TDelta(V_3)}}\rho_{x_1^nv_2^nv_3^n}^{Y_1}\11_{\{u^n = v_2^n\oplus v_3^n\}} & \leq 2^{n(H(V_2,V_3|U) + \delta_{u_2})}  \sum_{\substack{v_2^n\in\TDelta(V_2),\\ v_3^n \in \TDelta(V_3)}}p_{V_2V_3|U}^n(v_2^n,v_3^n|u^n)\rho_{x_1^nv_2^nv_3^n}^{Y_1}\11_{\{u^n = v_2^n\oplus v_3^n\}} \nonumber \\
& \leq 2^{n(H(V_2,V_3|U) + \delta_{u_2})}  \sum_{\substack{v_2^n,v_3^n}}p_{V_2V_3|U}^n(v_2^n,v_3^n|u^n)\rho_{x_1^nv_2^nv_3^n}^{Y_1}\11_{\{u^n = v_2^n\oplus v_3^n\}} \nonumber \\
& = 2^{n(H(V_2,V_3|U) + \delta_{u_2})} \rho_{x_1^nu^n}^{Y_1}.
\end{align}
The equality in \eqref{eq:T21label4} follows from the definition of $\rho_{u^n}^{Y_1} \deq \sum_{x^n \in \CalX^n}p_{X_1}^n(x_1^n)\rho_{x_1^n u^n}^{Y_1}$, and the inequality in \eqref{eq:T21label5} uses the property of conditional projector i.e., $ \pi_{u^n}\rho_{u^n}^{Y_1}\pi_{u^n} \leq 2^{-n(S(Y_1|U)_{\sigma_1} - \delta_{u_3})}\pi_{u^n},$ for $\sigma_{1}$ as defined in the statement of the theorem. Finally \eqref{eq:T21label6} follows using the cyclicity of trace and the fact that $\pi_{\hat{x}_1^n}\pi_{\rho} \pi_{u^n}\pi_{\rho}\pi_{\hat{x}_1^n} \leq \pi_{\hat{x}_1^n}\pi_{\rho}\pi_{\hat{x}_1^n} \leq  \pi_{\hat{x}_1^n} \leq I.$ 

\noindent Using the dimensional bound of a conditional typical projector i.e., $\tr(\pi_{\hat{x}_1^n,u^n}) \leq 2^{n(S(Y_1|X_1,U) + \delta_{u_4})} $ in \eqref{eq:T21Ana} we obtain
\begin{align}
    \EE[T_{21}] & \leq \frac{2^{nR_1}q^{2k}}{q^{2n}}2^{n(H(V_2,V_3|U) + \delta_{u_2})}2^{-n(S(Y_1|U)_{\sigma_1} - \delta_{u_3})}2^{n(S(Y_1|X_1,U) + \delta_{u_4})}2^{n(H(U)+\delta_{u_1})} \nonumber \\  
    & = \frac{2^{nR_1}q^{2k}}{q^{2n}}2^{n(H(V_2,V_3) + \delta_{u_2})}2^{-n(I(Y_1;X_1|U)_{\sigma_1} - \delta_{u_3})}2^{n\delta_{u_4}}2^{n\delta_{u_1}}. \nonumber
\end{align}
This completes the proof.

\subsection{Proof of Proposition \ref{prop:LemmaT_22}}
\label{appx:proofofPropT_22}
We begin by substituting the simplification performed in \eqref{eq:SimplifyRho} into the expression corresponding to $T_{22}.$ This gives
\begin{align}
    T_{22} & = \sum_{\substack{a'\neq a, l'\neq l}}\tr(\gamma_{m_1}^{a',l'} \pial\rho_{c,\ulinem}^{Y_1}\pial) \nonumber \\
    & = \sum_{\substack{a'\neq a, l'\neq l}}\sum_{x_1^n \in \CalX_1^n}\sum_{v_2^n, v_3^n \in \CalF_q^n}\sum_{a_2,a_3 \in \CalF_q^k}\tr(\gamma_{m_1}^{a',l'} \pial  \rho_{x_1^nv_2^nv_3^n}^{Y_1} \pial)\11_{\{x_1^n(m_1) = x_1^n\}}\11_{\CalJ} \nonumber \\
    & \leq \sum_{\substack{a'\neq a,\\ l'\neq l}}\sum_{x_1^n \in \CalX_1^n}\sum_{\substack{v_2^n \in \TDelta(V_2), \\ v_3^n \in \TDelta(V_3)}}\sum_{a_2,a_3 \in \CalF_q^k}\sum_{u^n, \hat{u}^n \in \CalF_q^n}\hspace{-10pt}\tr(\gammaXUhat \pi_{u^n} \rho_{x_1^nv_2^nv_3^n}^{Y_1} \pi_{u^n})\11_{\{x_1^n(m_1) = x_1^n\}}\nonumber \\ 
    & \hspace{3in} \11_{\{u^n = v_2^n\oplus v_3^n\}}\11_{\{\hat{u}^n = a'g_{I} + l'g_{O/I}+b_2^n+b_3^n\}}\11_{\CalK} \nonumber
\end{align}
Taking expectation over the codebook generation distribution gives
\begin{align}\label{eq:ExpecT22}
    \EE[T_{22}] &\leq \sum_{\substack{a'\neq a, l'\neq l}}\sum_{x_1^n \in \CalX_1^n}\sum_{\substack{v_2^n \in \TDelta(V_2), \\ v_3^n \in \TDelta(V_3)}}\sum_{a_2,a_3 \in \CalF_q^k}\sum_{u^n, \hat{u}^n \in \CalF_q^n}\tr(\gammaXUhat \pi_{u^n} \rho_{x_1^nv_2^nv_3^n}^{Y_1} \pi_{u^n})p_{X_1}^n(x_1^n)\11_{\{u^n = v_2^n\oplus v_3^n\}}\nonumber \\ 
    & \hspace{30pt} \PP(\hat{u}^n = a'G_{I} + l'G_{O/I}+B_2^n+B_3^n, v_2^n = a_2G_{I} + m_2 G_{O/I} + B_2^n, v_3^n = a_3G_{I} + m_3 G_{O/I} + B_3^n ) \nonumber\\
    & \labelrel\leq{eq:T22label1} \frac{q^{3k}q^{l}}{q^{3n}} \sum_{x_1^n \in \CalX_1^n}p_{X_1}^n(x_1^n)\sum_{\substack{v_2^n \in \TDelta(V_2), \\ v_3^n \in \TDelta(V_3)}}\sum_{u^n \in \CalF_q^n}\tr(\Big(\sum_{ \hat{u}^n \in \CalF_q^n}\gammaXUhat\Big) \pi_{u^n} \rho_{x_1^nv_2^nv_3^n}^{Y_1} \pi_{u^n})\11_{\{u^n = v_2^n\oplus v_3^n\}},
\end{align}
where the inequality \eqref{eq:T22label1} is obtained by noting that for $a' \neq a_2\oplus a_3$ and $ l' \neq m_1 \oplus m_2,$ we have 
\begin{align}\label{eq:prob3terms}
     \PP(\hat{u}^n = a'G_{I} + l'G_{O/I}+B_2^n+B_3^n, v_2^n = a_2G_{I} + m_2 G_{O/I} + B_2^n, v_3^n = a_3G_{I} + m_3 G_{O/I} + B_3^n ) = \frac{1}{q^{3n}}.  
\end{align}
Now consider the following simplification of $\sum_{ \hat{u}^n \in \CalF_q^n}\gammaXUhat.$ We have 
\begin{align}
    \sum_{ \hat{u}^n \in \CalF_q^n}\gammaXUhat & = \sum_{ \hat{u}^n \in \TDelta(U)} \pi_{\rho} \pi_{x_1^n} \pi_{x_1^n,u^n} \pi_{x_1^n}\pi_{\rho} \nonumber \\
    & \labelrel\leq{eq:gammalabel1} 2^{n(S(Y_1|X_1,U)_{\sigma_1} + \delta_{u_4})}\sum_{ \hat{u}^n \in \TDelta(U)} \pi_{\rho} \pi_{x_1^n} \rho_{x_1^nu^n} \pi_{x_1^n}\pi_{\rho} \nonumber \\
    & \labelrel\leq{eq:gammalabel2} 2^{n(S(Y_1|X_1,U)_{\sigma_1} + \delta_{u_4})}2^{n(H(U) + \delta_{u_1})}\sum_{ \hat{u}^n \in \TDelta(U)}p_{U}^n(u^n) \pi_{\rho} \pi_{x_1^n} \rho_{x_1^nu^n} \pi_{x_1^n}\pi_{\rho} \nonumber \\
    &\leq 2^{n(S(Y_1|X_1,U)_{\sigma_1} + \delta_{u_4})}2^{n(H(U) + \delta_{u_1})} \pi_{\rho} \pi_{x_1^n}\left(\sum_{ \hat{u}^n \in \CalF_q^n} p_{U}^n(u^n)\rho_{x_1^nu^n}\right) \pi_{x_1^n}\pi_{\rho} \nonumber \\
    & = 2^{n(S(Y_1|X_1,U)_{\sigma_1} + \delta_{u_4})}2^{n(H(U) + \delta_{u_1})} \pi_{\rho} \pi_{x_1^n}\rho_{x_1^n} \pi_{x_1^n}\pi_{\rho} \nonumber \\
    & \labelrel\leq{eq:gammalabel3} 2^{n(S(Y_1|X_1,U)_{\sigma_1} + \delta_{u_4})}2^{n(H(U) + \delta_{u_1})} 2^{-n(S(Y_1|X_1) - \delta_{x_1})} \pi_{\rho} \pi_{x_1^n}\pi_{\rho} \labelrel={eq:gammalabel4} c_1 \cdot I, \nonumber 
\end{align}
where \eqref{eq:gammalabel1} follows using the following arguments
\begin{align}
    \pi_{x_1^n,u^n} \leq 2^{n(S(Y_1|X_1,U)_{\sigma_1} + \delta_{u_4})}\pi_{x_1^n,u^n}\rho_{x_1^nu^n}\pi_{x_1^n,u^n}   & = 2^{n(S(Y_1|X_1,U)_{\sigma_1} + \delta_{u_4})}\sqrt{\pi_{x_1^n,u^n}}\rho_{x_1^nu^n}\sqrt{\pi_{x_1^n,u^n}} \nonumber \\
    & = 2^{n(S(Y_1|X_1,U)_{\sigma_1} + \delta_{u_4})}{\rho_{x_1^nu^n}}. \nonumber 
\end{align}
The inequalities \eqref{eq:gammalabel2},  \eqref{eq:gammalabel3} uses the typicality arguments. Lastly, the inequality
\eqref{eq:gammalabel4} follows by using the fact that $\pi_{\rho} \pi_{x_1^n}\pi_{\rho} \leq I$ and by defining $c_1 \deq 2^{n(S(Y_1|X_1,U)_{\sigma_1} +H(U) - S(Y_1|X_1) + \delta_{u_1} + \delta_{u_4}+ \delta_{x_1})}.$ 

Substituting the above simplification in \eqref{eq:ExpecT22}, we obtain
\begin{align}
    \EE[T_{22}] &\leq \frac{c_1q^{3k}q^{l}}{q^{3n}} \sum_{x_1^n \in \CalX_1^n}p_{X_1}^n(x_1^n)\sum_{\substack{v_2^n \in \TDelta(V_2), \\ v_3^n \in \TDelta(V_3)}}\sum_{u^n \in \CalF_q^n}\tr( \pi_{u^n} \rho_{x_1^nv_2^nv_3^n}^{Y_1} \pi_{u^n})\11_{\{u^n = v_2^n\oplus v_3^n\}}\nonumber \\
     &\labelrel\leq{eq:ExpecT22label1} \frac{c_1q^{3k}q^{l}}{q^{3n}} \sum_{\substack{v_2^n \in \TDelta(V_2), \\ v_3^n \in \TDelta(V_3)}}\sum_{u^n \in \CalF_q^n}\11_{\{u^n = v_2^n\oplus v_3^n\}} \nonumber \\
     &\labelrel\leq{eq:ExpecT22label2} \frac{q^{3k}q^{l}}{q^{3n}}2^{n( +H(U) - I(Y_1;U|X_1)_{\sigma_1} + \delta_{u_1} + \delta_{u_4}+ \delta_{x_1})} 2^{n(H(V_1)+H(V_2) + \delta_{v_1} + \delta_{v_2})}, \nonumber 
\end{align}
where \eqref{eq:ExpecT22label1} uses $\tr( \pi_{u^n} \rho_{x_1^nv_2^nv_3^n}^{Y_1} \pi_{u^n}) \leq 1$, and \eqref{eq:ExpecT22label1} follows from the definition of $c_1$ for $\sigma_1$ as defined in the statement of the theorem. This gives the desired rate to bound $\EE[T_{22}].$

\subsection{Proof of Proposition \ref{prop:LemmaT_23}}\label{appx:proofofPropT_23}
Using the simplification performed in \eqref{eq:SimplifyRho}, we obtain
\begin{align}
    T_{23} & = \sum_{\substack{m_1' \neq m_1,\\a'\neq a, l'\neq l}}\sum_{x_1^n \in \CalX_1^n}\sum_{v_2^n, v_3^n \in \CalF_q^n}\sum_{a_2,a_3 \in \CalF_q^k}\tr(\gamma_{m_1'}^{a',l'} \pial \rho_{x_1^nv_2^nv_3^n}^{Y_1} \pial)\11_{\{x_1^n(m_1) = x_1^n\}}\11_{\CalJ} \nonumber \\
    & \leq \sum_{\substack{m_1' \neq m_1,\\a'\neq a, l'\neq l}}\sum_{x_1^n,  \hat{x}_1^n \in \CalX_1^n}\sum_{\substack{v_2^n \in \TDelta(V_2), \\v_3^n \in \TDelta(V_3)}}\sum_{a_2,a_3 \in \CalF_q^k}\sum_{u^n, \hat{u}^n \in \CalF_q^n}\hspace{-10pt}\tr(\gammaXhatUhat \pi_{u^n} \rho_{x_1^nv_2^nv_3^n}^{Y_1} \pi_{u^n})\11_{\{x_1^n(m_1) = x_1^n\}} \nonumber \\ 
    &  \hspace{200pt}\11_{\{x_1^n(m_1') = \hat{x}_1^n\}} \11_{\{u^n = v_2^n\oplus v_3^n\}}\11_{\{\hat{u}^n = a'g_{I} + l'g_{O/I}+b_2^n+b_3^n\}}\11_{\CalK} ,
\end{align}
where the above inequality follows by noting that $\CalJ \subseteq \CalK.$ By taking expectation of the above term with respect to the codebook generating distributions, we get
\begin{align}\label{eq:ExpecT23}
    \EE[T_{23}] & \leq \sum_{\substack{m_1' \neq m_1,\\a'\neq a, l'\neq l}}\sum_{x_1^n, \hat{x}_1^n \in \CalX_1^n}\sum_{\substack{v_2^n \in \TDelta(V_2), \\v_3^n \in \TDelta(V_3)}}\sum_{a_2,a_3 \in \CalF_q^k}\sum_{u^n, \hat{u}^n \in \CalF_q^n}\hspace{-10pt}\tr(\gammaXhatUhat \pi_{u^n} \rho_{x_1^nv_2^nv_3^n}^{Y_1} \pi_{u^n})p_{X_1}^n(x_1^n)p_{X_1}^n(\hat{x}_1^n)\11_{\{u^n = v_2^n\oplus v_3^n\}} \nonumber \\ 
    & \hspace{30pt} \PP(\hat{u}^n = a'G_{I} + l'G_{O/I}+B_2^n+B_3^n, v_2^n = a_2G_{I} + m_2 G_{O/I} + B_2^n, v_3^n = a_3G_{I} + m_3 G_{O/I} + B_3^n ) \nonumber \\
    & \labelrel\leq{eq:ExpecT23label1} \frac{2^{nR_1}q^{3k}q^{l}}{q^{3n}}\sum_{x_1^n, \hat{x}_1^n \in \CalX_1^n} p_{X_1}^n(\hat{x}_1^n)p_{X_1}^n(x_1^n) \sum_{u^n, \hat{u}^n \in \CalF_q^n}\hspace{-10pt}\tr\Bigg(\gammaXhatUhat \pi_{u^n} \Bigg(\sum_{\substack{v_2^n \in \TDelta(V_2), \\v_3^n \in \TDelta(V_3)}}\rho_{x_1^nv_2^nv_3^n}^{Y_1}\Bigg) \pi_{u^n}\Bigg)\11_{\{u^n = v_2^n\oplus v_3^n\}}, \nonumber 
\end{align}
where the second inequality above uses the claim from \eqref{eq:prob3terms}.

Consider the following simplifications.
\begin{align}
    \sum_{x_1^n \in \CalX_1^n}p_{X_1}^n(\hat{x}_1^n)& \sum_{u^n, \hat{u}^n \in \CalF_q^n}\hspace{-10pt}\tr\Bigg(\gammaXhatUhat \pi_{u^n} \Bigg(\sum_{\substack{v_2^n \in \TDelta(V_2), \\v_3^n \in \TDelta(V_3)}}\rho_{x_1^nv_2^nv_3^n}^{Y_1} \11_{\{u^n = v_2^n\oplus v_3^n\}}\Bigg) \pi_{u^n}\Bigg)\nonumber \\
    & \labelrel\leq{eq:subSimplifyT22label1} 2^{n(H(V_2,V_3|U) + \delta_{u_2})}  \sum_{u^n, \hat{u}^n \in \TDelta(U)}\hspace{-10pt}\tr(\gammaXhatUhat \pi_{u^n} \Bigg(\sum_{x_1^n \in \CalX_1^n}p_{X_1}^n(\hat{x}_1^n)\rho_{x_1^nu^n}^{Y_1}\Bigg) \pi_{u^n})\nonumber \\
    & \labelrel\leq{eq:subSimplifyT22label2} 2^{n(H(V_2,V_3|U) + \delta_{u_2})}  \sum_{u^n, \hat{u}^n \in \TDelta(U)}\hspace{-10pt}\tr(\gammaXhatUhat \rho_{u^n}^{Y_1} )\nonumber \\
    & \leq  2^{n(H(V_2,V_3|U) + \delta_{u_2})} 2^{n(H(U) + \delta_{u_1})}  \sum_{\hat{u}^n \in \TDelta(U)}\hspace{-10pt}\tr(\pi_{\hat{x}_1^n}\pi_{\hat{x}_1^n,\hat{u}^n}\pi_{\hat{x}_1^n}\pi_{\rho}  \Big(\sum_{u^n \in\TDelta(U)}p_U^n(u^n)\rho_{u^n}^{Y_1}\Big)\pi_{\rho} ) \nonumber \\
    & \labelrel\leq{eq:subSimplifyT22label3} 2^{n(H(V_2,V_3|U) + \delta_{u_2})} 2^{n(H(U) + \delta_{u_1})}2^{-n(S(Y_1)_{\sigma} - \delta_{\rho})}  \sum_{\hat{u}^n \in \TDelta(U)}\hspace{-10pt}\tr( \pi_{\hat{x}_1^n,\hat{u}^n}),
\end{align}
where \eqref{eq:subSimplifyT22label1} holds from the above bounds obtained in \eqref{eq:simplifyRhoV1V2}, \eqref{eq:subSimplifyT22label2} follows by the definition of $\rho_{u^n}^{Y_1}$ and by using the inequality $ \pi_{u^n}\rho_{u^n}^{Y_1} \pi_{u^n} \leq \rho_{u^n}^{Y_1},$ \eqref{eq:subSimplifyT22label3} follows using (i) $\pi_{\rho}  \Big(\sum_{u^n \in\TDelta(U)}p_U^n(u^n)\rho_{u^n}^{Y_1}\Big)\pi_{\rho} \leq \pi_{\rho}  \rho \pi_{\rho} \leq 2^{-n(S(Y_1)_{\sigma} - \delta_{\rho})}\pi_{\rho} $ and (ii) $\tr(\pi_{\hat{x}_1^n}\pi_{\hat{x}_1^n,\hat{u}^n}\pi_{\hat{x}_1^n}\pi_{\rho}) \leq \tr(\pi_{\hat{x}_1^n,\hat{u}^n}). $ Using the above simplification in \eqref{eq:ExpecT23} we obtain
\begin{align}
    \EE[T_{23}] & \leq \frac{2^{nR_1}q^{3k}q^{l}}{q^{3n}} 2^{n(H(V_2,V_3|U) + \delta_{u_2})} 2^{n(H(U) + \delta_{u_1})}2^{-n(S(Y_1)_{\sigma} - \delta_{\rho})} \sum_{ \hat{x}_1^n \in \CalX_1^n} p_{X_1}^n(\hat{x}_1^n) \sum_{\hat{u}^n \in \TDelta(U)}\hspace{-10pt}\tr( \pi_{\hat{x}_1^n,\hat{u}^n}  )\nonumber \\ 
    & \leq  \frac{2^{nR_1}q^{3k}q^{l}}{q^{3n}} 2^{n(H(V_2,V_3|U) + \delta_{u_2})} 2^{n(H(U) + \delta_{u_1})}2^{-n(S(Y_1)_{\sigma} - \delta_{\rho})} 2^{n(H(U) + \delta_{u_1})}2^{n(S(Y_1|X_1,U)_{\sigma_1} + \delta_{u_4})}\nonumber \\ 
    & = \frac{2^{nR_1}q^{3k}q^{l}}{q^{3n}} 2^{n(H(V_2,V_3) + H(U) - I(Y_1;X_1,U)_{\sigma_1} + \delta_{\rho} + 2\delta_{u_1}+\delta_{u_2} + \delta_{u_4})} .
\end{align}
This completes the proof.

\bibliographystyle{IEEEtran}
{\bibliography{ThreeUserCQIC}}

\end{document}